\documentclass[letterpaper,10pt,journal]{IEEEtran} 

\usepackage{graphics}   % for pdf, bitmapped graphics files
\usepackage{epsfig}     % for postscript graphics 
\usepackage{mathptmx}   % assumes new font selection scheme installed
\usepackage{times}      % assumes new font selection scheme installed
\usepackage{amsmath}    % assumes amsmath package installed
\usepackage{amssymb}    % assumes amsmath package installed
\usepackage{amsthm}     % For theorems
\usepackage{xcolor}
\usepackage{enumerate}
\usepackage[caption=false,font=footnotesize]{subfig}
\usepackage{dsfont}

\newtheorem{prop}{Proposition}%[section]
\newtheorem{thm}{Theorem}%[section]
\newtheorem{defi}{Definition}%[section]
\newtheorem{cor}{Corollary}%[section]
\newtheorem{lem}{Lemma}%[section]
\newtheorem{ass}{Assumption}%[section]
\newtheorem{remark}{Remark}%[section]

\newcommand{\quotes}[1]{``#1''}
\newcommand{\prt}[1]{\left(#1\right)}
\newcommand{\brc}[1]{\left\{#1\right\}}
\newcommand{\brk}[1]{\left[#1\right]}
\newcommand{\bigprt}[1]{\big(#1\big)}
\newcommand{\bigbrk}[1]{\Big[#1\Big]}

\newcommand{\E}{\mathbb{E}}
\newcommand{\Ep}[1]{\E\brk{#1}}
\newcommand{\condE}[2]{\E\brk{#1|#2}}
\newcommand{\condP}[2]{P\brk{#1|#2}}

\newcommand{\lr}{\lambda_r}
\newcommand{\lc}{\lambda_c}
\newcommand{\lp}{\lambda_p}
\newcommand{\rratio}{\rho}          
\newcommand{\lclr}{\tilde \rho} 
\newcommand{\lrlc}{\eta}         

\newcommand{\Ping}{Ping}
\newcommand{\Gossip}{Gossip}

\newcommand{\Kset}[2]{\omega_{#1}^{#2}}
\newcommand{\Ksetie}{\Kset{i}{\epsilon^*}}
\newcommand{\PingKset}[1]{\Kset{#1}{\Ping}}

\newcommand{\GossipKset}[1]{\Kset{#1}{\Gossip}}

\newcommand{\Tji}[1]{T_{j\rightarrow i}^{#1}}
\newcommand{\Tjie}{\Tji{\epsilon^*}}
\newcommand{\Fjit}[1]{F_{j\rightarrow i}^{t,#1}}
\newcommand{\fjit}[1]{f_{j\rightarrow i}^{t,#1}}
\newcommand{\PingTji}{\Tji{\epsilon^{\Ping}}}
\newcommand{\PingFjit}{\Fjit{\Ping}}
\newcommand{\Pingfjit}{\fjit{\Ping}}
\newcommand{\GossipTji}{\Tji{\epsilon^{\Gossip}}}
\newcommand{\GossipFjit}{\Fjit{\Gossip}}
\newcommand{\Gossipfjit}{\fjit{\Gossip}}

\newcommand{\ExistInfo}[2]{\bar N_{#1}^{#2}}
\newcommand{\NoInfo}[2]{N_{#1}^{#2}}
\newcommand{\InfoExch}[2]{{#1\to#2}}

\newcommand{\yi}[1]{y_i^{#1}}
\newcommand{\Cji}[1]{C_{j\rightarrow i}^{#1}}

\newcommand{\R}{\mathbb R}

\newcommand{\ONE}{\mathds{1}}

\newcommand{\Melem}[3]{\brk{#1}_{#2,#3}}
\newcommand{\Velem}[2]{\brk{#1}_{#2}}
\newcommand{\vect}[1]{\mathbf{#1}}

\title{
Fundamental Performance Limitations for\\ Average Consensus in Open Multi-Agent Systems
}

\author{Charles~Monnoyer~de~Galland and Julien~M.~Hendrickx% 
\thanks{ICTEAM institute, UCLouvain (Belgium). This work was supported by \quotes{Communaut\'e fran\c{c}aise de Belgique - Actions de Recherche Concert\'ees}. C. M. is a FRIA fellow (F.R.S.-FNRS). This work is supported by the ``\textit{RevealFlight}'' ARC at UCLouvain. Email adresses: \texttt{charles.monnoyer@uclouvain.be}, \texttt{julien.hendrickx@uclouvain.be}.
}%
}

\date{January 2021}

\begin{document}

\maketitle

%%%%%%%%%%%%%%%%%%%%%%%%%%%%%%%%%%%%%%%%%%%%%%%%%%%%%%%%%%%%%%%%%%%%
%%%%%%%%%%%%%%%%%%%%%%%%%%%%%%%%%%%%%%%%%%%%%%%%%%%%%%%%%%%%%%%%%%%%
\begin{abstract}
We derive \textit{fundamental performance limitations} for intrinsic average consensus problems in open multi-agent systems, which are systems subject to frequent arrivals and departures of agents.
Each agent holds a value, and the objective of the agents is to collaboratively estimate the average of the values of the agents presently in the system.
Algorithms solving such problems in open systems are poised to never converge because of the permanent variations in the composition, size and objective pursued by the agents of the system.
We provide lower bounds on the expected Mean Squared Error achievable by any averaging algorithms in open systems of fixed size.
Our derivation is based on the analysis of a conceptual algorithm that would achieve optimal performance for a given model of replacements.
We obtain a general bound that depends on the properties of the model defining the interactions between the agents, and instantiate that result for all-to-one and one-to-one interaction models.
A comparison between those bounds and algorithms implementable with those models is then provided to highlight their validity.
\end{abstract}

%%%%%%%%%%%%%%%%%%%%%%%%%%%%%%%%%%%%%%%%%%%%%%%%%%%%%%%%%%%%%%%%%%%%
%%%%%%%%%%%%%%%%%%%%%%%%%%%%%%%%%%%%%%%%%%%%%%%%%%%%%%%%%%%%%%%%%%%%
\section{Introduction}
\label{Sec:Intro}
Multi-agent systems are a powerful tool that allows modelling and solving a wide variety of problems in various domains.
Those include learning or tracking in sensor networks \cite{ApMAS:WSN,ApMAS:TrackingWSN}, vehicle coordination \cite{ApMAS:vehicle} or formation control through consensus \cite{ApMas:Consensus_based_formation_control:CDC2019}, social phenomena such as the study of opinion dynamics \cite{ApMAS:CTAvgOpinionDynamics,ApMAS:Opinion&Confidence}, and many more.

Most results stated around multi-agent systems stand for asymptotic properties and assume that their composition remains unchanged throughout the whole process, contrasting with the fact that some of the most desired features of such systems are scalability and flexibility.
This apparent contradiction is justified if the process and the arrivals and departures of agents take place at different time scales: the analyses then stand for a period of time where no such event occurred.
However, as the system size grows, both the characteristic length of the process and the probability for agents to leave or join the system tend to increase, so that arrivals and departures become significant in the analysis of large systems.
Think for instance of birth rates in living systems, which are expected to grow with the system size. 
Moreover, the constant composition assumption is also challenged by some processes where communications are difficult or infrequent by nature or due to extreme environments: those thus take place at a time scale comparable to that of arrivals and departures.
Some systems are also inherently subject to arrivals and departures, such as \textit{e.g.} multi-vehicles systems on a stretch of road where vehicles permanently join and leave.

We call \quotes{\textit{Open multi-agent systems}} systems subject to arrivals and departures of agents.
Such configurations raise many challenges as those repeated arrivals and departures imply important differences in analyses and algorithm design.
First, each arrival and departure has an impact on both the size and the state of the system, which prevents it from converging.
It thus becomes challenging to analyse its evolution, and new tools are needed as a consequence.
Second, the design of algorithms is also impacted by arrivals and departures, as those can also modify the objective that is pursued by the agents.
Additionally, it may be necessary for algorithms to deal with either outdated information that must be erased if the desired output of the algorithm is defined by the agents presently in the system; or in opposition to deal with losses of information due to departures, and that should be remembered.
Finally, some results cannot be easily extended to open systems, such as \textit{e.g.}, the analysis of Gossip interactions and the design of MAX-consensus algorithms, respectively studied in open systems in \cite{OMAS:Gossip:DeterministicAllerton,OMAS:Gossip:RandomCDC} and \cite{OMAS:MAX}.

%%%%%%%%%%%%%%%%%%%%%%%%%%%%%%%%%%%%%%%%%%%%%%%%%%%%%%%%%%%%%%%%%%%%
\subsection{State of the art}

Arrivals and departures have been considered in different applications, such as the THOMAS architecture designed to maintain connectivity into P2P networks \cite{ApOMAS:OpenP2P}, or with VTL for autonomous cars to deal with cross-sections \cite{ApOMAS:VTL}. 
Analyses on social phenomena subject to arrivals and departures were performed in \cite{OMAS:sociophysics}, but those are mostly simulation based.
More generally, the Plug and Play implementation is getting attention for designing structures on systems where subsystems can be plugged or unplugged \cite{OMAS:PnP:DecentralizedMPC,OMAS:PnP:PartisionBasedDistrKalman}.

Different settings have been considered to model and study open multi-agent systems.
A first model consists in having a finite number of \quotes{potentially present} agents, which are either in or out of the system at any given time. This implies that leaving agents may come back in the system, and eventually every arrival concerns a returning agent.
The sensitivity of consensus algorithms to noise in open systems was studied in \cite{OMAS:ImpactOfNoise_RandomConsensusAlgo} in this setting.
In an alternative model, that we follow here, every arriving agent is a new one, and leaving agents never come back.
The problem of convergence in open system was studied in this context, respectively in \cite{OMAS:Median_consensus:CDC2019} for the median consensus which is shown to be tracked with bounded error under some conditions, and in \cite{OMAS:Gossip:DeterministicAllerton,OMAS:Gossip:RandomCDC} where Gossip interactions are studied through the analysis of size-independent descriptors.
Open dynamic consensus towards the average and the maximum were analyzed in \cite{OMAS:OpenDynamicConsensus:Franseschelli-Frasca,OMAS:DynamicMaxMinSizeEstimation} using that model.
Algorithm design was also explored for MAX-consensus in open systems using additional variables in \cite{OMAS:MAX} by investigating the probability of convergence if the system eventually closes.
Finally, decentralized optimization in this setting is receiving attention, such as in \cite{DO:online-varyingFunctions} where time-varying objective functions are considered, and in \cite{OpenDo:OpenDGDStability} where the stability of decentralized gradient descent is analyzed.

However, efficiently studying the performance of algorithms, and the analysis of open systems in general, remains a challenge.
Since agents cannot instantaneously react to perturbations, and since the objective they pursue may be varying, converging to an exact result is not achievable.
An alternative way to measure the performance of algorithms consists in ensuring the error remains within some range, as the usual convergence is no more a relevant criterion.

%%%%%%%%%%%%%%%%%%%%%%%%%%%%%%%%%%%%%%%%%%%%%%%%%%%%%%%%%%%%%%%%%%%%
\subsection{Contribution and preliminary results}
We consider the average consensus problem in open system, \textit{i.e.}, the collaborative estimation of the average of the values owned by the agents \textit{presently in the system at some time}, under arrivals and departures of agents.
We assume agents can exchange information with each other in order to build their estimate of the average, and that the system size is fixed, so that each departure of an agent is instantaneously followed by an arrival (we call such event \quotes{replacement}).
We establish fundamental performance limitations for that problem, \textit{i.e.}, lower bounds on the expected performance of any algorithm.

Our contributions are as follows.
We provide in Section~\ref{Sec:General} a general lower bound on the expected performance of any algorithm that can be implemented with a wide variety of interaction schemes defining the information exchanges in the system.
Its derivation relies on the notion of \emph{knowledge set}, that we introduce in Section~\ref{Sec:Ksets}, and which aims at modelling the information made potentially available to an agent through information exchanges.
We then particularize our bound to two particular interaction schemes that allow the implementation of pairwise Gossip interactions \cite{Avg:Gossip}, and compare the results with their actual performance in Section~\ref{Sec:App}.

In a preliminary version of this work \cite{OMAS:CDC2019:FPL_intrAVG}, we assumed the system had been running since $-\infty$, and derived lower bounds on the expected asymptotic performance of averaging algorithms.
Our main new contribution as compared to \cite{OMAS:CDC2019:FPL_intrAVG} are
(i) the derivation of tighter bounds for pairwise interaction rules as compared to the Infection model of \cite{OMAS:CDC2019:FPL_intrAVG};
(ii) the generalization of our result to lower bounds on the performance at all time, including transient performance, as opposed to only asymptotic performance;
(iii) the use of knowledge sets, that allow for simpler and clearer proofs and leads to more general results extendable to any interaction model.

This analysis represents a first step towards understanding open systems, as those limitations can be used as a quality criterion for algorithm design in open systems, and highlight the main bottlenecks that could arise in their analysis.
More generally, Gossip interactions and consensus in general are used as intermediate steps in many applications, including decentralized optimization \cite{OpenDo:OpenDGDStability,DO:dualAvg} where optimization steps and consensus steps alternate, or estimation over wireless networks \cite{ApMAS:TrackingWSN} and decentralized control \cite{ApMAS:vehicle,ApMas:Consensus_based_formation_control:CDC2019} relying on achieving consensus on the different measurements or on the distance between vehicles.
Hence, we expect the techniques we use to be extendable to more advanced tasks in open systems, and to be a starting point for their analysis.

%%%%%%%%%%%%%%%%%%%%%%%%%%%%%%%%%%%%%%%%%%%%%%%%%%%%%%%%%%%%%%%%%%%%
%%%%%%%%%%%%%%%%%%%%%%%%%%%%%%%%%%%%%%%%%%%%%%%%%%%%%%%%%%%%%%%%%%%%
\section{Notations}
\label{Sec:Notations}
We use $P(A)$ to denote the probability of an event $A$, and $\condP{A}{B}$ its probability conditional to the event $B$.
We denote by $\bar A$ the complementary event of $A$, so that $P(\bar{A})=1-P(A)$.

We denote by $F_X(x) = P\brk{X\leq x}$ the \textit{Cumulative Distribution Function} (CDF) of a random variable $X$, and $f_X(x) = \frac{d}{dx}F_X(x)$ its \textit{Probability Density Function} (PDF).
By definition, $\lim_{x\to\infty} F_X(x)=1$; we will introduce in Section~\ref{Sec:General} the \quotes{pseudo-CDF} of $X$ that lifts this condition in order to ease the representation of events with a positive probability of never occurring.
We use $\Ep{X}$ to denote the expected value of $X$, and $\condE{X}{Y}$ its expected value conditional to $Y$.

We denote by $\alpha(t^-)$ and $\alpha(t^+)$ the value of a time-varying signal or process $\alpha$ respectively just before and after an event modifying it at time $t$.

We use $\Velem{\vect v}{k}$ to denote the $k$-th element of a vector $\vect v$ and $\Melem{M}{i}{j}$ the element of the $i$-th row and $j$-th column of matrix $M$.
Moreover, $\ONE$ denotes the vector constituted only of ones, and $\mathbf {e_k}$ the vector such that $\brk{\mathbf{e_k}}_{k}=1$, and $\brk{\mathbf{e_k}}_{i}=0$ for all $i\neq k$.

%%%%%%%%%%%%%%%%%%%%%%%%%%%%%%%%%%%%%%%%%%%%%%%%%%%%%%%%%%%%%%%%%%%%
%%%%%%%%%%%%%%%%%%%%%%%%%%%%%%%%%%%%%%%%%%%%%%%%%%%%%%%%%%%%%%%%%%%%
\section{Problem statement}
\label{Sec:Stt}

\subsection{System description}
\label{sec:Stt:System}

We consider a system initialized at time $t=0$ with $N$ agents labelled from $1$ to $N$, subject to arrivals and departures of agents.
Let $S(t)$ denote the set of the labels of the agents in the system at time $t$, which satisfies thus $S(0)=\brc{1,\ldots,N}$, and evolves as agents join and leave the system. 
In particular, every arrival at some time $t$ results in $S(t^+) = S(t^-)\cup\brc{i'}$ with $i'\notin S(\tau)$, $\forall \tau\leq t$, and the departure of an agent $i'\in S(t^-)$ at time $t$ results in $S(t^+) = S(t^-)\setminus\brc{i'}$.
Each agent $i$ is attributed an i.i.d. constant value $x_i$ when it enters the system.
This value is randomly drawn from a distribution whose mean is assumed to be a known constant, set to zero without loss of generality, and with variance $\E x_i^2=\sigma^2$.
In this work, for the sake of simplicity, we assume the system is subject to only two types of events -- replacements and information exchanges -- and its size is thus constant.

\paragraph{Replacements}
Each agent $i$ that leaves the system is immediately replaced by a new agent $i'$, so that \textit{the system size remains constant}.
The agent $i'$ is given a new i.i.d. random value $x_i'$ drawn from the same distribution as that of agent $i$.

\paragraph{Information exchanges}
Agents can exchange information with each other at discrete times, and we denote by $\InfoExch ji$ the event that agent $j$ sends information to agent $i$ through such exchange.
There is no limit to what agent $j$ can send, so that it can potentially send everything it knows.
The occurrence of those exchanges can be defined according to different models, \textit{e.g.}, in a all-to-one or a pairwise manner (respectively considered in Sections \ref{Sec:App:Ping} and \ref{Sec:App:SIS}).

The goal of the agents in the system is to collaboratively estimate the average of the values held by the agents \emph{presently in the system at time $t$}, defined as 
\begin{equation}
    \label{eq:Stt:System:Average}
    \bar x(t) = \frac{1}{|S(t)|}\sum_{i\in S(t)}x_i.
\end{equation}
\vspace{-0.5cm}
\subsection{Problem reformulation}
We directly introduce a relaxation that allows simplifying the system description: we assume that an agent joining the system at a replacement knows the label of the one it replaces (or equivalently inherits it).
Hence, the replacement of an agent $i$ can be equivalently formulated as (i) \emph{the attribution of a new value $x_i'$ to that agent} and (ii) \emph{the erasure of all other variables it holds}.
The labels of the agents in the system thus remain $\brc{1,\ldots,N}$ at all times, and the value held by an agent $i$ becomes a piecewise constant value $x_i(t)$ that remains constant between replacements.
We note a major difference between this formulation and (closed) systems where the signal $x_i(t)$ held by each agent $i$ is time-varying, as from (ii) each variation of $x_i(t)$ (\textit{i.e.}, replacement of $i$) is coupled with the erasure of all the information $i$ holds.

With this new formulation, the average \eqref{eq:Stt:System:Average} becomes
\begin{equation}
    \label{eq:Stt:System:Average:reformulation}
    \bar x(t) = \tfrac1N\sum\nolimits_{i=1}^N x_i(t).
\end{equation}
The results we derive with this relaxation are also valid for algorithms that do not make use of it.
Observe this reformulation is not directly feasible if the system size changes.

%%%%%%%%%%%%%%%%%%%%%%%%%%%%%%%%%%%%%%%%%%%%%%%%%%%%%%%%%%%%%%
\subsection{Objective and relaxations}
\label{sec:Stt:Relaxations}

The usual convergence of the estimate $y_i(t)$ of an agent $i$ is not achievable in our open system, due to the time-varying nature of the objective $\bar x(t)$, and to the time it takes for the agents to react to perturbations.
Hence, in order to measure the performance of algorithms that compute such estimate in our setting, we make the natural choice to analyze the Mean Squared Error, defined as follows:
\begin{equation}
    \label{eq:STT:MSE_C(t)}
    C(t) := \frac1N \sum_{i=1}^N \prt{\bar x(t) - y_i(t)}^2.
\end{equation}

Our goal is to derive lower bounds on the expected value of the criterion above which are valid for any algorithm implementable in our setting.

For that purpose, we consider a setting such that agents have access to more information than what is typically allowed in open multi-agent systems: the size of the system $N$, the labels of the agents, the distribution defining their values, and the dynamics of the system, \textit{i.e.}, the stochastic properties of the replacements and information exchanges.
Moreover, we assume the agents have access to a universal time, and that they have unlimited memory, so that they can potentially send everything they know during interactions.
We analyze an algorithm that is provably optimal in that setting, so that its performance is a lower bound on that of any other implementable algorithm.
Notice that our bounds are also valid for any setting allowing for less information to be accessible to the agents, or subject to more constraints.

%%%%%%%%%%%%%%%%%%%%%%%%%%%%%%%%%%%%%%%%%%%%%%%%%%%%%%%%%%%%%%%%%%%%
%%%%%%%%%%%%%%%%%%%%%%%%%%%%%%%%%%%%%%%%%%%%%%%%%%%%%%%%%%%%%%%%%%%%
\section{Knowledge sets}
\label{Sec:Ksets}

\subsection{Definitions}
We now provide a general representation of the set containing all the information that an agent \textit{may have access to} from a sequence of replacements and information exchanges.
We name it \textit{knowledge set}.
In our setting, the information contained in that set can be of two types: it is either related to the value held by an agent at some time instant, or to the occurrence of an information exchange at some time.

\begin{defi}[Knowledge set]
    \label{def:Stt:General_Kset}
    Let $\epsilon$ be a sequence of replacements and information exchanges.
    The \emph{knowledge set} of agent $i$ at time $t$ is the set $\Kset{i}{\epsilon}(t)$ initialized at the arrival of agent $i$ in the system at time $t_0$, either at the beginning of the process or when it is replaced, as
    \begin{equation}
        \Kset{i}{\epsilon}(t_0^+) = \brc{\prt{x_i(t_0),i,t_0}},
    \end{equation}
    and that is updated in the event $\InfoExch ji$ at time $\tilde t$ as
    \begin{equation}
        \Kset{i}{\epsilon}(\tilde t^+) = \Kset{i}{\epsilon}(\tilde t^-) \cup \Kset{j}{\epsilon}(\tilde t^-) \cup \brc{\prt{x_j(\tilde t),j,\tilde t}} \cup \brc{\prt{\InfoExch ji,\tilde t}},
    \end{equation}
    where $\prt{x_j(\tilde t),j,\tilde t}$ and $(\InfoExch ji,\tilde t)$ respectively denote the information that the value of agent $j$ at time $\tilde t$ is $x_j(\tilde t)$ and that the event $\InfoExch ji$ happened at time $\tilde t$.
\end{defi}

The knowledge set thus represents all the information that could be made available to an agent $i$ through a sequence of events.
At its arrival in the system, agent $i$ only knows itself.
Then each time it receives information from an agent $j$, its knowledge set records everything $j$ knows at that time, and keeps track of that information exchange (that is the last term of the union).
Information exchanges could be given a label that would also be recorded in the knowledge set at interactions; this is omitted since it is not used for obtaining our results.

\begin{remark}
    One could use other definitions similar to Definition~\ref{def:Stt:General_Kset} to handle different types of replacements and information exchanges, such as \textit{e.g.}, replacements with memory inheritance where agents do not forget what they know at replacements, or broadcast information exchanges where agents send their information to all other agents instead of only one.
\end{remark}

A graphical representation of the evolution of knowledge sets as presented in Definition~\ref{def:Stt:General_Kset} is provided in Fig.~\ref{fig:stt:KsetExample} for an event sequence constituted of two information exchanges followed by a replacement.
Standard results in computer science guarantee that the result of any (deterministic) algorithm can be computed based solely on knowledge sets as they are defined above, see \textit{e.g.}, \cite{MISC:ViewsInAGraph_Part1,MISC:ViewsInAGraph_Julien}.

\begin{figure}
    \centering
    \includegraphics[width=0.45\textwidth]{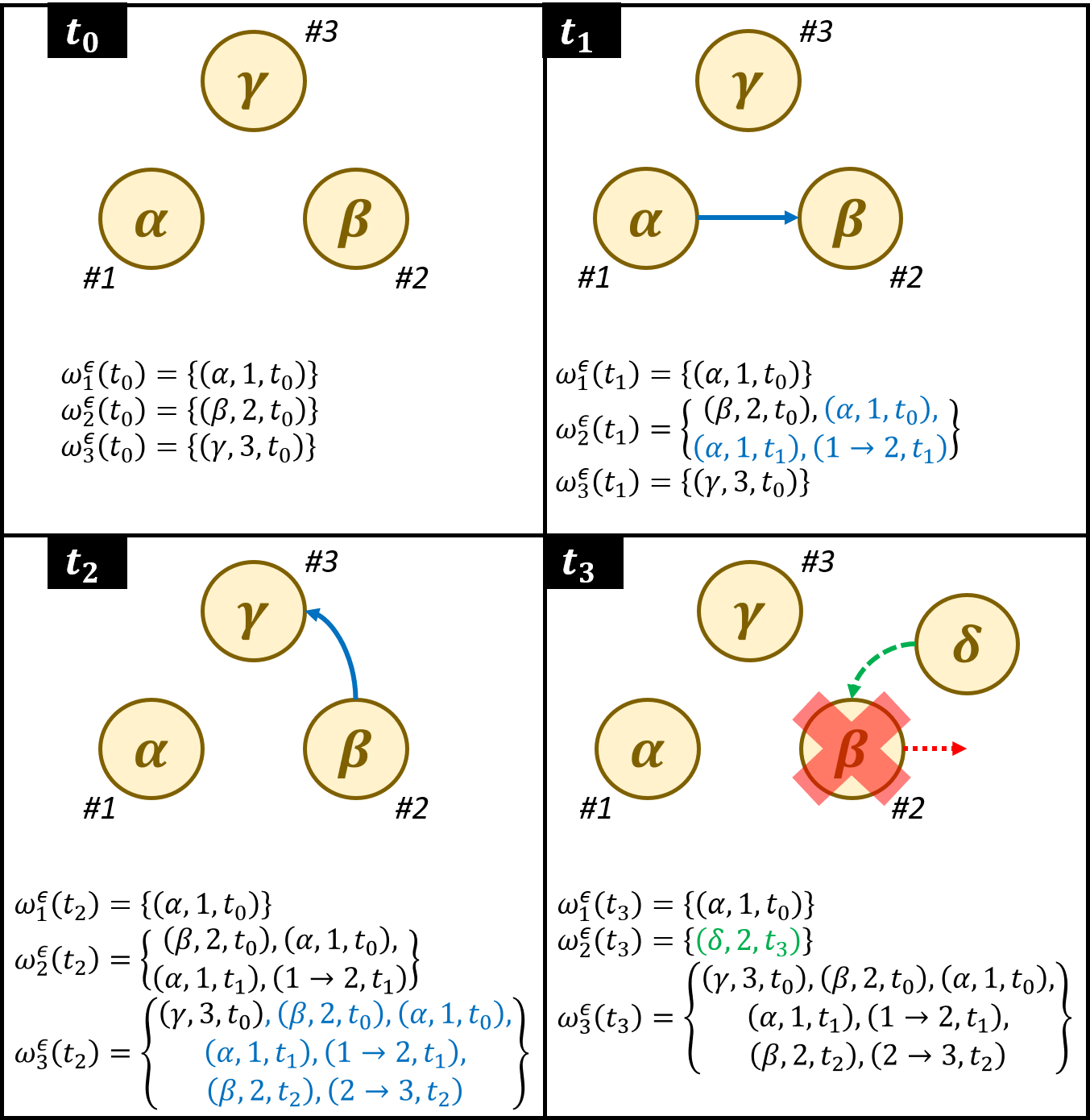}
    \caption{Evolution of the knowledge sets of three agents as defined in Definition~\ref{def:Stt:General_Kset}, in an open system as described in Section~\ref{Sec:Stt}, subject to two consecutive information exchanges (from agent 1 to agent 2, and then from agent 2 to agent 3), and then to a replacement of agent 2. Each box corresponds to the state of the knowledge sets of the agents right after each event, \textit{i.e.}, at initialization, at the interactions and at the replacement.}
    \label{fig:stt:KsetExample}
\end{figure}

The following lemma directly follows Definition~\ref{def:Stt:General_Kset}, and means that anything that can be computed from a given sequence of replacements and information exchanges can also be computed if more exchanges are performed.

\begin{lem}
    \label{lem:Stt:Kset:Kset_Inclusion}
    Let $\epsilon$ and $\epsilon'$ be two sequences of replacements and information exchanges such that $\epsilon\subseteq\epsilon'$ with the same replacement events, then it follows for all $i$ and all $t$ that
    \begin{equation}
        \label{eq:lem:Stt:Kset:Kset_Inclusion}
        \Kset{i}{\epsilon}(t) \subseteq \Kset{i}{\epsilon'}(t),
    \end{equation}
    and hence any algorithm that can be computed on $\Kset{i}{\epsilon}(t)$ can also be computed on $\Kset{i}{\epsilon'}(t)$.
\end{lem}

\begin{proof}
    The proof directly follows the fact that all the information gathered in $\Kset{i}{\epsilon}(t)$ with $\epsilon$ is also obtained in $\Kset{i}{\epsilon'}(t)$ with $\epsilon'$ through the same information exchanges by construction.
\end{proof}

We now define the \emph{age of the most recent information about an agent $j$ that is available to agent $i$}, that will be used to quantitatively assess the relevance of an information in $\Kset{i}{\epsilon}(t)$.

\begin{defi}[Age of the most recent information]
    \label{def:Stt:AgeOfInfo}
    Given a knowledge set $\Kset{i}{\epsilon}(t)$, the \emph{age of the most recent information about agent $j$ available to agent $i$ at time $t$} (if it exists) is defined as
    \begin{equation}
        \label{eq:def:Stt:AgeOfInfo}
        \Tji{\epsilon}(t) := \min\brc{s:\prt{x_j(t-s),j,t-s}\in\Kset{i}{\epsilon}(t)}.
    \end{equation}
\end{defi}

Observe that $\Tji\epsilon(t)$ as defined above is the amount of time spent since the time instant corresponding to the most recent information about $j$ held by $i$, and thus only exists if there is information about agent $j$ in $\Kset{i}{\epsilon}(t)$.

%%%%%%%%%%%%%%%%%%%%%%%%%%%%%%%%%%%%%%%%%%%%%%%%%%%%%%%%%%%%%%%%%
\subsection{Properties for stochastic event sequences}

From now on, we will consider stochastic event sequences $\epsilon^*$, satisfying the following assumption.

\begin{ass}
\label{ass:stochasticEventSeq}
Let $\epsilon^*$ be a sequence of events constituted of replacements and information exchanges.
\begin{enumerate}
    \item Individual replacements follow a Poisson clock of rate $\lambda_r$, and are independent of each other, of information exchanges and of the values held by the agents;
    \item Information exchanges happen according to some stochastic interaction model denoted $*$. They are independent of replacements and of the values held by the agents (but not necessarily of each other).  
\end{enumerate}
\end{ass}

Assumption~\ref{ass:stochasticEventSeq} implies that on average $N\lambda_r$ replacements are expected to take place in the whole system per unit of time.
We use Poisson clocks because they are standard for modelling memoryless events.
Moreover, we will see that the relevance of information held by an agent $i$ about another agent $j$ (\textit{i.e.}, the probability estimated by $i$ that $j$ is still present in the system) is then entirely determined by the age of the most recent information.
The two following lemmas state useful properties of knowledge sets for stochastic event sequences satisfying Assumption~\ref{ass:stochasticEventSeq}, that we will use to derive or results.

\begin{lem}
    \label{lem:Stt:Kset:AgeOfInfo}
    Consider a stochastic event sequence $\epsilon^*$ satisfying Assumption~\ref{ass:stochasticEventSeq}, and let $R_j[s,t]$ denote the event that agent $j$ was replaced between times $s$ and $t$, then for all times $s$ such that $\prt{x_j(s),j,s}\in\Kset{i}{\epsilon^*}(t)$, there holds
    
    \begin{small}
    \begin{equation}
        \label{eq:lem:Stt:Kset:AgeOfInfo}
        \condP{R_j[s,t]}{\Kset{i}{\epsilon^*}(t)}
        = \begin{cases}
            1                   &\hbox{if }R_j[s,t-\Tjie(t)]\\
            1-e^{-\lr\Tjie(t)}  &\hbox{otherwise}
        \end{cases}.
    \end{equation}
    \end{small}
\end{lem}

The above lemma means that in the event $R_j[s,t-\Tjie(t)]$, \textit{i.e.}, if agent $j$ was replaced between times $s$ and $t-\Tjie(t)$, then agent $i$ knows it for sure from the information it acquired at time $t-\Tjie(t)$ (by definition of $\Tjie(t)$).
Hence, in that case, the probability that agent $j$ was replaced between times $s$ and $t$ provided the knowledge of agent $i$ at time $t$, denoted $\condP{R_j[s,t]}{\Kset{i}{\epsilon^*}(t)}$, is $1$.
Otherwise, this probability depends on $\Tjie(t)$, following the Poisson process governing replacements.
Therefore, the relevance of an information about agent $j$ in $\Kset{i}{\epsilon^*}(t)$ is entirely characterized by $\Tjie(t)$, and does not depend on the labels of the agents nor the values they hold.
It thus prevents situations where the replacement probability can be assessed from something else than the age (\textit{e.g.}, if agents were to inform everyone at their arrival), or where information can be neglected or deleted because of the value or label of the agents.
Moreover, it provides a proper expression for the replacement probability.

\begin{lem}
\label{lem:IndepKsetValue}
    Let $t_0$ be an arrival time of agent $i$ in the system (at initialization or through a replacement), then for all times $s\leq t_0$, and for any agent $j$, $x_i(t_0^+)$ and $\Kset{j}{\epsilon^\ast}(s)$ are independent.
\end{lem}

Lemma~\ref{lem:IndepKsetValue} means that at the moment agent $i$ is replaced, its new value is independent of any prior information in the system, and it is impossible for other agents to estimate that value.
In particular, no information about that value can be deduced from previous values, interaction times, or even from the absence of such information.

%%%%%%%%%%%%%%%%%%%%%%%%%%%%%%%%%%%%%%%%%%%%%%%%%%%%%%%%%%%%%%%%%%%%
%%%%%%%%%%%%%%%%%%%%%%%%%%%%%%%%%%%%%%%%%%%%%%%%%%%%%%%%%%%%%%%%%%%%
\section{General lower bound}
\label{Sec:General}
In this section, we obtain a generic lower bound on the expected Mean Squared Error given by $\Ep{C(t)}$, that is valid for any algorithm implementable with a given interaction model governing the information exchanges, denoted $*$.

Consider a stochastic event sequence $\epsilon^*$, where the information exchanges are obtained from the model $*$.
Then $\Kset{i}{\epsilon^*}(t)$, the knowledge set of agent $i$ at time $t$ obtained from $\epsilon^*$, is also stochastic.
Moreover, $\Tjie(t)$, the age of the most recent information about an agent $j$ in $\Ksetie(t)$ as defined in Definition~\ref{def:Stt:AgeOfInfo}, is a random variable that is well defined only if there exists information about agent $j$ in $\Ksetie(t)$.
Hence, we generalize the notion of \emph{Cumulative Density Function (CDF)} \cite{MISC:Statistics}, and we call the \quotes{\emph{pseudo-CDF}} of $\Tjie(t)$ the function
\begin{equation}
    \label{eq:General:PseudoCDF}
    \Fjit*(s) = P\brk{\Tjie(t)\leq s}
\end{equation}
that allows $\lim_{s\to\infty}\Fjit*(s)\neq1$ if there exists no information about agent $j$ in $\Ksetie(t)$ (denote that event $\NoInfo{j}{i}(t)$).
We similarly call \quotes{\emph{pseudo-PDF}} of $\Tjie(t)$ the function 
\begin{equation}
    \label{eq:General:PseudoPDF}
    \fjit*(s) = \frac{d}{ds}\Fjit*(s).
\end{equation}
The definition of pseudo-PDF is actually linked with that of \quotes{conditional probability distributions} \cite{MISC:Statistics}.
Denote (with an abuse of notation) $F_{T|\bar N}(s) = \condP{\Tjie(t)\leq s}{\ExistInfo{j}{i}(t)}$.
This function is a CDF, and there holds $\Fjit*(s) = F_{T|\bar N}(s)P(\ExistInfo{j}{i}(t))$.
Hence, $\fjit*(s)$ exists if the pdf $\frac{d}{ds}F_{T|\bar N}(s)$ exists, which is guaranteed as long as $F_{T|N}(s)$, and thus $\Fjit*(s)$, is (absolutely) continuous, which we assume to be the case in our work.

\begin{remark}
    \label{rem:pseudoPDF}
    The pseudo-PDF $\fjit*(s)$ captures the way information travels in the system.
    In particular, it reflects the properties of the interaction model (\textit{e.g.}, broadcast to all neighbours, pairwise interactions etc.) and any restriction on the communication, such as the network topology.
\end{remark}

We now provide the generic lower bound on $\Ep{C(t)}$, that is valid for any interaction model $*$ in the following theorem.

\begin{thm}
\label{thm:General:bound}
For any algorithm that can be computed on a knowledge set obtained from an interaction model $*$ under Assumption~\ref{ass:stochasticEventSeq}, there holds
\begin{equation}
    \label{eq:thm:General:bound}
    \Ep{C(t)} 
    \geq \frac{1}{N^3} \sum_{i=1}^N \sum_{j=1}^N \prt{1-\int_0^t \fjit{*}(s)e^{-2\lr s}\ \mathrm ds} \sigma^2,
\end{equation}
where $C(t)$ refers to the MSE defined in equation (\ref{eq:STT:MSE_C(t)}), and where $\fjit*(s)$ is the pseudo-PDF as defined in \eqref{eq:General:PseudoPDF} for $*$.
This holds even if agents know the system size $N$, the replacement rate $\lr$, the distribution defining their values, and $\fjit*(s)$.
\end{thm}

The result above is also valid for any instantiation of $\fjit*(s)$ with another pseudo-PDF that bounds it (\textit{i.e.}, such that it defines a random variable smaller than $\Tjie(t)$ in the usual stochastic order \cite{MISC:UsualStochasticOrder}).
See the preliminary version of this work \cite{OMAS:CDC2019:FPL_intrAVG} for an example of application of this.

Corollary~\ref{cor:General:symmetric_steadyState} particularizes equation (\ref{eq:thm:General:bound}) for interaction models that do not make any distinctions between the agents, and where we assume the system has been running for a long time.

\begin{cor}
    \label{cor:General:symmetric_steadyState}
    If $\fjit*(s) = f^*(s)$ for some $f^*(s)$ holds $\forall i,j$ and for all time $t$, then there holds
    \begin{equation}
        \label{eq:cor:General:symmetric_steadyState:symmetry}
        \Ep{C(t)} \geq \frac{N-1}{N^2} \prt{1-\int_0^t f^*(s)\ e^{-2\lambda_rs}\mathrm ds} \sigma^2.
    \end{equation}
    In addition, when $t\rightarrow \infty$, there holds
    \begin{equation}
        \label{eq:cor:General:symmetric_steadyState:symmetry_steadyState}
        \lim\inf_{t\rightarrow\infty} \Ep{C(t)} \geq \frac{N-1}{N^2} \prt{1-\int_0^\infty f^*(s)\ e^{-2\lambda_rs}\mathrm ds} \sigma^2.
    \end{equation}
\end{cor}

The use of $\lim\inf_{t\rightarrow\infty}$ in Corollary~\ref{cor:General:symmetric_steadyState} is necessary because the limit might not exist as $t\to\infty$.

The remainder of this section provides intermediate results that will allow us to build step by step the proof of Theorem~\ref{thm:General:bound}, that is formally provided at the end of this section.

%%%%%%%%%%%%%%%%%%%%%%%%%%%%%%%%%%%%%%%%%%%%%%%%%%%%%%%%%%%%%%%%%%%%%%%%%%%%%
\subsection{Optimal estimate description}
\label{sec:General:opti_algo}
We begin by showing that an optimal estimate is obtained by choosing the expected value of the average provided the information in the knowledge set.
Hence, the performance of any algorithm implementable on this knowledge set is by definition lower bounded by that of this estimate. 
An algorithm computing this estimate would achieve optimal performance, but may not be implementable with reasonable resources.

\begin{prop}
\label{prop:General:opti_algo}
    Let $\Ksetie(t)$ be the knowledge set obtained from an interaction model $*$, and let us define
    \begin{equation}
        \label{eq:prop:General:opti_algo}
        y_i^*(t) := \condE{\bar x(t)}{\Ksetie(t)}.
    \end{equation}
    Then, for any algorithm implementable on a knowledge set $\Ksetie(t)$ obtained from the interaction model $*$, there holds
    \begin{equation}
        \label{eq:prop:General:optimality}
        \Ep{C(t)} \geq \Ep{\prt{\bar x(t)-y_i^*(t)}^2}.
    \end{equation}
\end{prop}

Therefore, we refer to $y_i^*(t)$ as the \quotes{\emph{optimal estimate}}.

\begin{proof}
    Using the fact that any estimate $y$ is deterministic conditional to $\Ksetie(t)$, one shows that
    \begin{equation*}
        y_i^*(t) = \arg\min_{y} \brc{\condE{\prt{\bar x(t)-y}^2}{\Ksetie(t)}},
    \end{equation*}
    Hence, the result of any other algorithm implementable on $\Ksetie(t)$ satisfies by definition
    \begin{equation*}
        \condE{C(t)}{\Ksetie(t)} \geq \condE{\prt{\bar x(t) - y_i^*(t)}^2}{\Ksetie(t)}.
    \end{equation*}
    The relation above is true for any realization of $\Ksetie(t)$ obtained with $*$.
    Therefore, for any algorithm, there holds
    \vspace{-0.3cm}
    
    \begin{small}
    \begin{align*}
        \Ep{C(t)} &= \Ep{\condE{C(t)}{\Ksetie(t)}}\\
        &\geq \Ep{\condE{\prt{\bar x(t) - y_i^*(t)}^2}{\Ksetie(t)}} = \Ep{\prt{\bar x(t) - y_i^*(t)}^2}
    \end{align*}
    \end{small}
    which concludes the proof.
\end{proof}

%%%%%%%%%%%%%%%%%%%%%%%%%%%%%%%%%%%%%%%%%%%%%%%%%%%%%%%%%%%%%%%%%%%%%%%%%%%%%
\subsection{Error decomposition into independent terms}
\label{sec:General:decomp}

The next proposition builds on the description of the optimal estimate in Proposition~\ref{prop:General:opti_algo} to highlight a convenient property of the system, which allows reducing the analysis of $\Ep{C(t)}$ to that of the MSE of estimation of a single agent's value.

\begin{prop}
    \label{prop:General:decomp}
    Let $\Ksetie(t)$ be the knowledge set obtained from an interaction model $*$, then with the optimal estimate \eqref{eq:prop:General:opti_algo} the expected MSE reduces to
    \begin{equation}
        \label{eq:prop:General:decomp}
        \Ep{C(t)} 
        = \frac{1}{N^3}\sum_{i=1}^N \sum_{j=1}^N \Ep{\prt{x_j(t) - \condE{x_j(t)}{\Ksetie(t)}}^2}.
    \end{equation}
\end{prop}

\begin{proof}
    The estimate (\ref{eq:prop:General:opti_algo}) can be rewritten as follows
    $$\yi{*}(t) = \condE{\bar x(t)}{\Ksetie(t)} = \frac{1}{N} \sum_{j=1}^N\condE{x_j(t)}{\Ksetie(t)}.$$
    \noindent Hence, it follows that
    \begin{align*}
        \Ep{C(t)} 
        &= \Ep{\frac1N\sum_{i=1}^N \prt{\bar{x}(t)-y_i^*(t)}^2}\\
        &= \frac1N\sum_{i=1}^N \prt{ 
            \frac{1}{N^2} \Ep{\sum_{j=1}^N \prt{x_j(t)-\condE{x_j(t)}{\Ksetie(t)}}}^2
            }.
    \end{align*}
    \noindent From the absence of correlation between the agents, the crossed-product terms of the sum above (\textit{i.e.}, for $i\neq j$) are nullified, and by denoting $\hat{x}_i^{(j)}(t) = \condE{x_j(t)}{\Ksetie(t)}$ there holds
    \begin{align*}
        \sum_{i=1}^N\Ep{\sum_{j=1}^N \prt{x_j(t)-\hat{x}_i^{(j)}(t)}}^2
        = \sum_{i=1}^N\Ep{\sum_{j=1}^N \prt{x_j(t)-\hat{x}_i^{(j)}(t)}^2},
    \end{align*}
    and the conclusion directly follows.
\end{proof}

%%%%%%%%%%%%%%%%%%%%%%%%%%%%%%%%%%%%%%%%%%%%%%%%%%%%%%%%%%%%%%%%%%%%%%%%%%%%%
\subsection{Single error analysis}
\label{sec:General:single_error}

From Proposition~\ref{prop:General:decomp}, we only need to analyze the MSE of estimation of a single agent's value to characterize $\Ep{C(t)}$.
Let us define the following criterion:
\begin{equation}
    \label{eq:General:single_error:singleMSE}
    \Cji{*}(t) := \prt{x_j(t)-\condE{x_j(t)}{\Ksetie(t)}}^2,
\end{equation}
where the dependence of $\Cji*(t)$ on $x_j(t)$ is omitted.
The expected value of $\Cji*(t)$ is the MSE we need to analyze.
We will need the following lemma, proved in Appendix~\ref{sec:Annex:proof_Lemma_MSE_RV}.

\begin{lem}
    \label{lem:General:single_error:MSE_RV}
    Let $Y,Z$ be two i.i.d. zero mean random variables of variance $\sigma^2$, and define another random variable $X$ as follows
    \begin{equation}
        \label{eq:lem:General:single_error:MSE_RV:rv_def}
        X := 
        \begin{cases}
            Z  &\hbox{with probability } p\\
            Y       &\hbox{with probability } 1-p
        \end{cases},
    \end{equation}
    such that the event related to the probability $p$ is independent of $Z$ and $Y$.
    The estimator $\hat X$ that minimizes $\Ep{\prt{X-\hat X}^2}$ provided the value of $Z$ is given by
    \begin{equation}
        \label{eq:lem:General:single_error:MSE_RV:estimate}
        \hat X
        = \arg\min_{x} \brc{\condE{\prt{X-x}^2}{Z}} = pZ,
    \end{equation}
    and the error is then given by
    \begin{equation}
        \label{eq:lem:General:single_error:MSE_RV:rv_error}
        \Ep{\prt{X - \hat X}^2} = (1-p^2)\sigma^2.
    \end{equation}
\end{lem}

The following proposition provides an expression for the expected value of $\Cji*(t)$.

\begin{prop}
    \label{prop:General:single_error}
    Let $\Ksetie(t)$ be the knowledge set obtained from an interaction model $*$, then there holds
    \begin{equation}
        \label{eq:prop:General:single_error}
        \Ep{\Cji{*}(t)} = \prt{1 - \int_0^t \fjit{*}(s)e^{-2\lr s}\mathrm{d}s}\sigma^2 
    \end{equation}
    where $\fjit{*}(s)$ is pseudo-PDF of $\Tjie(t)$ as defined in \eqref{eq:General:PseudoPDF}.
\end{prop}

\begin{proof}
    Let us characterize
    $$\Ep{\Cji{*}(t)} = \Ep{\prt{x_j(t)-\condE{x_j(t)}{\Ksetie(t)}}^2}.$$
    
    Let $\NoInfo{j}{i}(t)$ denote the event that there is no information about agent $j$ in $\Ksetie(t)$, \textit{i.e.}, the event that $\prt{x_j(s),j,s}\notin\Ksetie(t)$ for all $s\leq t$, and let $\ExistInfo{j}{i}(t)$ denote its complementary event.
    
    The random variable $\Tjie(t)$ is well defined only if there is such information (\textit{i.e.}, in the event $\ExistInfo{j}{i}(t)$).
    Hence there holds
    \begin{align*}
        \Ep{\Cji*(t)} = &P\brk{\NoInfo{j}{i}(t)} \condE{\Cji*(t)}{\NoInfo{j}{i}(t)} \\
        &+ \int_0^t \fjit*(s)\ \condE{\Cji*(t)}{\Tjie(t)=s}\mathrm ds,
    \end{align*}
    where we remind $\fjit*(s)$ is the pseudo-PDF of $\Tjie(t)$ as defined in \eqref{eq:General:PseudoPDF}.
    In the event $\NoInfo{j}{i}(t)$, then Lemma~\ref{lem:IndepKsetValue} guarantees that $\condE{x_j(t)}{\Ksetie(t)\cap\NoInfo{j}{i}(t)}=0$, and the MSE is given by
    \begin{equation*}
        \condE{\Cji*(t)}{\NoInfo{j}{i}(t)}
        = \Ep{x_j(t)^2}
        = \sigma^2.
    \end{equation*}
    
    In the event $\ExistInfo{j}{i}(t)$, the value $x_j(t)$ given $\Ksetie(t)$ is 
    \begin{equation*}
        x_j(t) = 
        \begin{cases}
            x_j\prt{t-\Tjie(t)}     &\hbox{w.p. } e^{-\lr\Tjie(t)}\\
            X                       &\hbox{otherwise}
        \end{cases},
    \end{equation*}
    where $X$ is the unknown random value taken by agent $j$ if a replacement happened, and is thus a random variable independent of $x_j\prt{t-\Tjie(t)}$ following the same distribution.
    The probability $e^{-\lr\Tjie(t)}$ is the probability that agent $j$ was not replaced since the time $t-\Tjie(t)$ given $\Ksetie(t)$, obtained from Lemma~\ref{lem:Stt:Kset:AgeOfInfo}.

    Lemma~\ref{lem:General:single_error:MSE_RV} can then be applied to obtain
    \begin{equation*}
        \condE{\Cji*(t)}{\Tjie(t)} = \prt{1-e^{-2\lr\Tjie(t)}}\sigma^2.
    \end{equation*}
    Moreover, the following holds by definition of $\Tjie(t)$:
    \begin{equation*}
        \int_0^t \fjit*(s)\mathrm ds = P\brk{\ExistInfo{j}{i}(t)} = 1-P\brk{\NoInfo{j}{i}(t)}.
    \end{equation*}
    Combining everything together then gives
    \begin{align*}
        &\Ep{\Cji{*}(t)} = P\brk{\NoInfo{j}{i}(t)} \sigma^2 + \int_0^t \fjit*(s) \prt{1-e^{-2\lr s}}\sigma^2\mathrm ds\\
        &= \prt{P\brk{\NoInfo{j}{i}(t)}+1-P\brk{\NoInfo{j}{i}(t)} - \int_0^t \fjit{*}(s) e^{-2\lr s} \mathrm ds}\sigma^2,
    \end{align*}
    which ultimately leads to the conclusion.
\end{proof}

Proposition~\ref{prop:General:single_error} is the last result we need to build the proof of Theorem~\ref{thm:General:bound}, that we provide now.
\begin{proof}[Proof of Theorem~\ref{thm:General:bound}]
    First, Proposition~\ref{prop:General:opti_algo} shows that the performance of any algorithm is lower bounded by that achieved by the estimate $y_i^*(t) = \condE{\bar x(t)}{\Ksetie(t)}$, so that there holds
    \begin{equation*}
        \Ep{C(t)} \geq \Ep{\prt{\bar x(t)-y_i^*(t)}^2}.
    \end{equation*}
    Second, Proposition~\ref{prop:General:decomp} shows that 
    \begin{equation*}
        \Ep{\prt{\bar x(t)-y_i^*(t)}^2} 
        = \frac{1}{N^3}\sum_{i=1}^N \sum_{j=1}^N \Ep{\Cji*(t)}
    \end{equation*}
    with 
    \begin{equation*}
        \Cji{*}(t) := \prt{x_j(t)-\condE{x_j(t)}{\Ksetie(t)}}^2.
    \end{equation*}
    Finally, there holds from Proposition~\ref{prop:General:single_error} that
    \begin{equation*}
        \Ep{\Cji{*}(t)} = \prt{1 - \int_0^t \fjit{*}(s)e^{-2\lr s}\mathrm{d}s}\sigma^2.
    \end{equation*}
    Injecting this in the result of Proposition~\ref{prop:General:decomp}, and combining it with that of Proposition~\ref{prop:General:opti_algo} thus concludes the proof.
\end{proof}

\begin{remark}
    \label{rem:General:InterModel_VS_algo}
    Even though the proof of Theorem~\ref{thm:General:bound} builds on the analysis of a particular algorithm (\textit{i.e.}, Proposition~\ref{prop:General:opti_algo}), its result is valid for any algorithm that can be implemented on the same interaction model, including unstable ones.
\end{remark}
In the sequel, we provide several instantiations of the result of Theorem~\ref{thm:General:bound} by considering interaction models that allow the implementation of the Gossip algorithm, and that will thus be valid lower bounds for its performance.
More generally, one can instantiate the result of Theorem~\ref{thm:General:bound} with any suitable interaction model by defining the corresponding pseudo-PDF $\fjit*(s)$, and injecting it into \eqref{eq:thm:General:bound}.

%%%%%%%%%%%%%%%%%%%%%%%%%%%%%%%%%%%%%%%%%%%%%%%%%%%%%%%%%%%%%%%%%%%%
%%%%%%%%%%%%%%%%%%%%%%%%%%%%%%%%%%%%%%%%%%%%%%%%%%%%%%%%%%%%%%%%%%%%
\section{Instantiations of the bound}
\label{Sec:App}

In this section, we instantiate the lower bound on $\Ep{C(t)}$ obtained in Theorem~\ref{thm:General:bound} for two interaction models.
The first one, for illustrative purpose, supposes that the agents can learn the current state of the whole system at once, and that no memory erasure is performed at replacements.
The second one consists in pairwise exchanges of information.
Note that the first model already provides a (conservative) lower bound for the setting of the second one: as it will be seen it allows in particular the implementation of the Gossip algorithm, defined on the second model.

%%%%%%%%%%%%%%%%%%%%%%%%%%%%%%%%%%%%%%%%%%%%%%%%%%%%%%%%%%%%%%%%%%%%
\subsection{First model: Ping updates}
\label{Sec:App:Ping}

\begin{defi}[Ping model]
\label{def:App:Ping:model}
    The Ping model satisfies:
    \begin{itemize}
        \item Each agent $i$ receives, at times defined by an individual Poisson clock of rate $\lambda_p$, information from all other agents, \emph{i.e.}, $\InfoExch{j}{i}$ for all $j$. We call this update \quotes{Ping update};
        \item At a replacement, the arriving agent inherits all the information held by the agent being replaced.
    \end{itemize}
    We denote by $\PingKset{i}(t)$ the knowledge set of agent $i$ at time $t$ corresponding to a stochastic event sequence $\epsilon^{Ping}$ obtained from that model.
\end{defi}
Note that the Ping model requires a slight adaptation of Definition~\ref{def:Stt:General_Kset} to allow memory inheritance at replacements, which would read $\PingKset{i}(t_R^+) = \PingKset{i}(t_R^-) \cup \brc{\prt{x_i',i,t_R}}$ with $t_R$ being a replacement time of agent $i$.
One can verify that all our results and arguments from Section~\ref{Sec:General} remain valid.
We provide in Fig.~\ref{fig:PingKset} an illustration of possible evolutions of the knowledge sets of three agents with the Ping model. 

Observe moreover that no assumption is made on the dependence or independence of the Poisson clocks driving the Ping updates of the different agents; the results happen to be independent of such assumptions (in that sense it is actually a class of models rather than a model). 
Introducing dependencies will reveal useful, as \textit{e.g.}, having pairs of agents performing simultaneous Ping updates provides a relaxation for pairwise interactions, considered in Section~\ref{Sec:App:SIS}.

\begin{figure}
    \centering
    \includegraphics[width=0.5\textwidth]{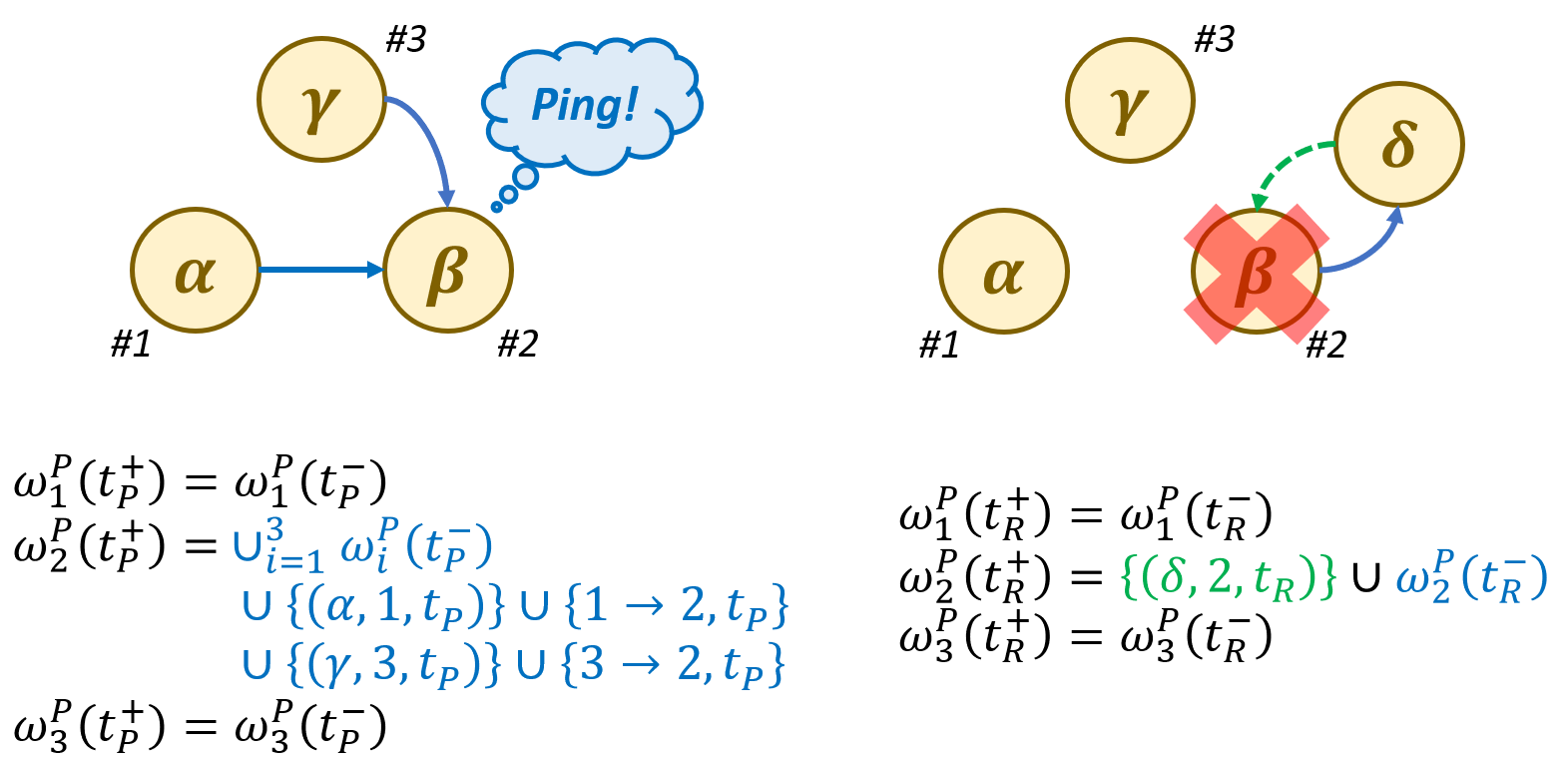}
    \caption{Impact of a Ping update and of a replacement on the knowledge sets with the Ping model as described in Section~\ref{Sec:App:Ping}. An agent obtains the information held by all the other agents at that time when it performs a Ping update (on the left), or an agent inherits the information held by the one it replaces at a replacement (on the right). We use the notation $\Kset{i}{P}(t)$ instead of $\PingKset{i}(t)$ to lighten the notations.}
    \label{fig:PingKset}
\end{figure}

To apply Theorem~\ref{thm:General:bound} and obtain our bound, we first need to compute the pseudo-PDF $\Pingfjit(s)$ as defined in \eqref{eq:General:PseudoPDF}.

\begin{prop}
    \label{prop:Ping:fji}
    With the Ping model, the random variable $\PingTji(t)$ admits the following pseudo-PDF for any $j\neq i$:
    \begin{equation}
        \label{eq:prop:Ping:fji}
        \Pingfjit(s) = \lp e^{-\lp s}.
    \end{equation}
\end{prop}

\begin{proof}
With the Ping model, the age of the most recent information about an agent held by another one is the age of the last time it performed a Ping update.
Since those happen according to a Poisson clock of rate $\lp$, there holds
\begin{equation*}
    \label{eq:proof:prop:Ping:fji}
    \PingFjit(s) = P\brk{\PingTji(t)\leq s} = 1-e^{-\lp s}.
\end{equation*}
The conclusion follows from $\Pingfjit(s) = \tfrac{d}{ds}\PingFjit(s)$.
\end{proof}

Since the pseudo-PDF $\Pingfjit(s)$ obtained in Proposition~\ref{prop:Ping:fji} does not depend either on the agents nor the time $t$, the lower bound on $\Ep{C(t)}$ can be derived from applying Corollary~\ref{cor:General:symmetric_steadyState}.

\begin{thm}
\label{thm:Ping:bound}
For any algorithm that can be implemented on $\PingKset{i}(t)$, there holds
\begin{equation}
    \label{eq:thm:Ping:bound}
    \Ep{C(t)} 
    \geq \frac{N-1}{N^2} 
    \prt{\frac{1}{1+\frac12\frac{\lp}{\lr}} + \frac{e^{-\prt{\lp+2\lr}t}}{1+2\frac{\lr}{\lp}}} \sigma^2.
\end{equation}
This holds even if agents know the system size $N$, the rates $\lr$ and $\lp$, and the distribution defining their values.
\end{thm}

\begin{proof}
The conclusion follows the instantiation of $f^{*}(s)$ in Corollary~\ref{cor:General:symmetric_steadyState} with $\Pingfjit(s)$ from Proposition~\ref{prop:Ping:fji}.
The bound is then obtained from the integration that follows:
\begin{equation*}
    \label{eq:proof:thm:Ping:bound}
    \Ep{C(t)} 
    \geq \frac{N-1}{N^2} \prt{1-\int_0^t \lp e^{-\lp s}e^{-2\lr s}\mathrm ds} \sigma^2.
\end{equation*}
The claim follows from a few algabraic manipulations after integrating the expression.
\end{proof}

Taking the limit as $t\rightarrow\infty$, or similarly by applying the second result of Corollary~\ref{cor:General:symmetric_steadyState}, we get the following corollary.

\begin{cor}
    \label{cor:Ping:symmetric_steadyState}
    For any algorithm that can be implemented on $\PingKset{i}(t)$, there holds
    \begin{equation}
        \label{eq:cor:Ping:symmetric_steadyState:symmetry_steadyState}
        \lim\inf_{t\rightarrow\infty} \Ep{C(t)} \geq \frac{N-1}{N^2} 
        \prt{\frac{1}{1+\frac12\frac{\lp}{\lr}}} \sigma^2.
    \end{equation}
\end{cor}

\begin{figure}
    \centering
    \includegraphics[width=0.5\textwidth,clip = true, trim=0.5cm 10.25cm 0cm 10.5cm,keepaspectratio]{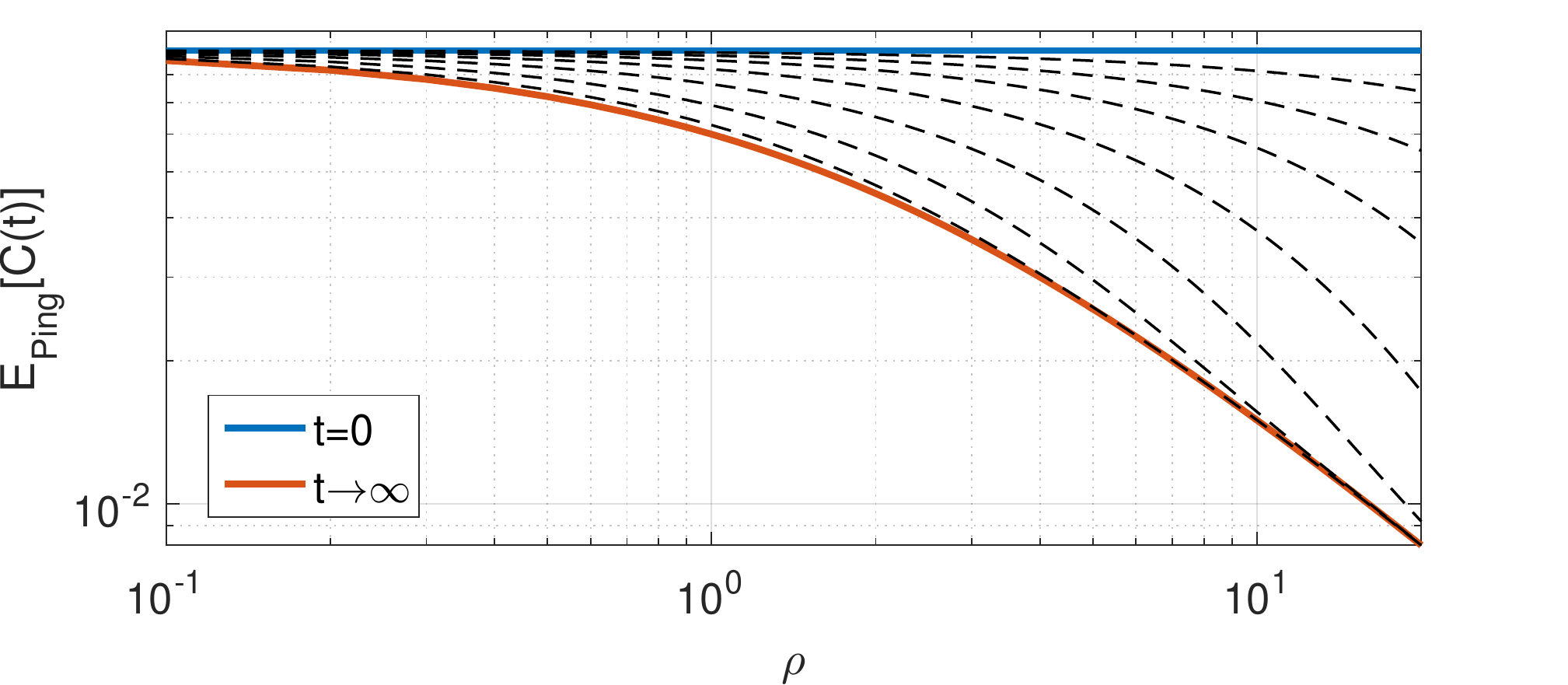}
    \caption{Time-dependent bound derived with the Ping model (\ref{eq:thm:Ping:bound}) with 10 agents and $\sigma^2=1$. The blue line is the MSE at the initialization of the system ($t=0$), the red line is the bound obtained when $t\rightarrow\infty$ (\ref{eq:cor:Ping:symmetric_steadyState:symmetry_steadyState}), and the dashed black lines are the bounds obtained at some times in between.} 
    \label{fig:Ping:TimeDepBound}
\end{figure}

Fig.~\ref{fig:Ping:TimeDepBound} shows the bound (\ref{eq:thm:Ping:bound}) in terms of the rate ratio $\rratio:=\lp/\lr$ (\textit{i.e.}, the average number of Ping updates experienced by an agent before leaving the system) for 10 agents at different times after initializing the system, until it converges when $t\rightarrow\infty$.
When the system is initialized, the only information available to the agents is their own value, and the MSE is exactly $\tfrac{N-1}{N^2}\sigma^2$ no matter the rate ratio.
With time passing, this MSE tends to be maintained if the communications are rather rare ($\rho\rightarrow0$), since the amount of available information remains mostly the same.
By contrast, the MSE decays to $0$ as the replacements become more frequent ($\rho\rightarrow\infty$), with the system behaving progressively like a closed system as $\rratio$ increases.
This behavior is enhanced as the time passes, to ultimately converge to a time-independent bound as $t\rightarrow\infty$.

Interestingly, the only dependence of the bound on the size of the system $N$ lies in the first factor of (\ref{eq:thm:Ping:bound}), $\frac{N-1}{N^2}\sigma^2$, namely the error at the initialization of the system.
The decay of that MSE is then entirely defined by the Ping update and replacement rates.
As $t\rightarrow\infty$, that decay gets only characterized by the rate ratio $\rratio = \lp/\lr$ \textit{i.e.}, the expected number of Ping updates experienced by an agent before leaving the system.

%%%%%%%%%%%%%%%%%%%%%%%%%%%%%%%%%%%%%%%%%%%%%%%%%%%%%%%%%%%%%%%%%%%%
\subsection{Second model: pairwise interactions}
\label{Sec:App:SIS}

\begin{defi}[Gossip model]
    \label{Def:App:SIS:model}
    The system is subject to strictly pairwise interactions: each pair of agents $(i,j)$ (with $i\neq j$) interacts at random times defined by a Poisson clock of rate $\frac{1}{N-1}\lambda_c$, resulting in the simultaneous occurrence of $\InfoExch{j}{i}$ and $\InfoExch{i}{j}$.
    We denote by $\GossipKset{i}(t)$ the knowledge set of agent $i$ at time $t$ corresponding to a stochastic event sequence $\epsilon^{Gossip}$ obtained from that model.
\end{defi}

With this model, which corresponds to defining pairwise undirected exchanges of information, each agent is expected to interact on average $\lambda_c$ times per unit of times.
We provide in Fig.~\ref{fig:GossipKset} an illustration of the evolution of the knowledge sets of three agents at an interaction with the Gossip model.

The Gossip model raises more challenges than the Ping model does. 
Additionally to handling outdated information from unknown replacements, memory losses are considered as well as the propagation time of information that is not instantaneously transmitted to all the agents at interactions.

\begin{figure}
    \centering
    \includegraphics[width=0.45\textwidth]{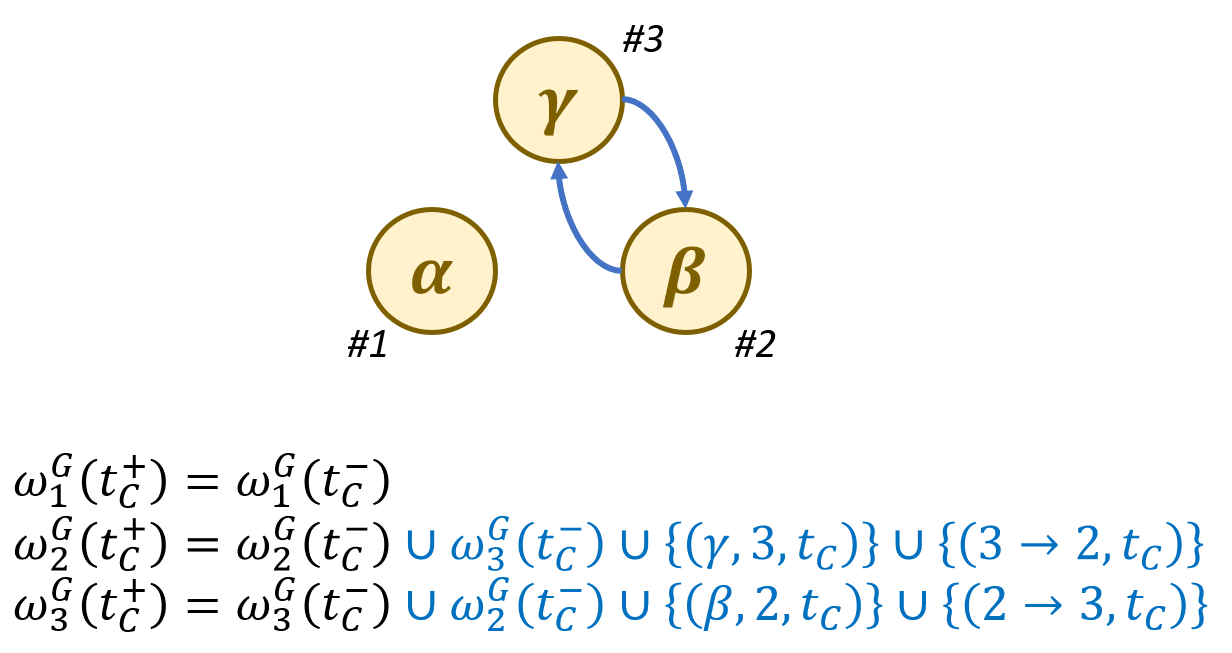}
    \caption{Occurrence of an information exchange with the Gossip model as described in Section~\ref{Sec:App:SIS}. When two agents $i$ and $j$ interact, both information exchanges from $j$ to $i$ and from $i$ to $j$ happen simultaneously. We use the notation $\Kset{i}{G}(t)$ instead of $\GossipKset{i}(t)$ to lighten the notations.}
    \label{fig:GossipKset}
\end{figure}

To apply Theorem~\ref{thm:General:bound} and obtain our bound, we first need to compute the pseudo-PDF $\Gossipfjit(s)$ as defined in \eqref{eq:General:PseudoPDF}.

\begin{prop}
    \label{prop:SIS:fji}
    With the Gossip model, the random variable $\GossipTji(t)$ admits the following pseudo-PDF for any $j\neq i$:
    \begin{equation}
        \label{eq:prop:SIS:fji}
        \Gossipfjit(s) = w^TAe^{As}\vect e_1,
    \end{equation}
    where $\Velem{w}{k} = \tfrac{k-1}{N-1}$, and where $A$ is a tridiagonal matrix with
    \begin{itemize}
        \item $\Melem{A}{k}{k} = -\prt{\frac{k(N-k)}{N-1}\lc +(k-1)\lr}$;
        \item $\Melem{A}{k}{k+1} = k\lr$;
        \item $\Melem{A}{k+1}{k} = \frac{k(N-k)}{N-1}\lc$.
    \end{itemize}
\end{prop}

\begin{proof}
    The pseudo-CDF $\GossipFjit(s)$ is defined as
    \begin{align*}
        \GossipFjit(s) 
        &= P\brk{\GossipTji(t)\leq s}\\
        &= P\brk{\exists \tau\in\brk{t-s,t}: \bigprt{j,\tau,x_j(\tau)}\in\GossipKset{i}(t)}.
    \end{align*}
    
    Denote by $n_j(t,t')$ the number of agents having at least one information about $j$ more recent than $t'$ at time $t$.
    There holds
    \begin{align*}
        \GossipFjit(s)
        &= \sum_{k=1}^N \frac{k-1}{N-1} P\brk{n_j(t,t-s)=k}.
    \end{align*}
    
    The factor $\frac{k-1}{N-1}$ comes from the absence of distinction between the agents during the interactions and replacements.
    
    The value $n_j(t,t-s)$ is constant between events, and potentially modified at replacements and interactions. 
    Indeed, if $n_j(t,t-s)=k$, then this value
    \begin{itemize}
        \item increases by one with rate $\frac{k(N-k)}{N-1}\lc$ from interactions;
        \item decreases by one with rate $(k-1)\lr$ from replacements.
    \end{itemize}
    
    The value $n_j(t,t-s)$ thus follows a continuous-time Markov chain, and more precisely a birth-death process.
    
    Denote $\vect P(s)$ such that $\Velem{\vect P(s)}{k} = P\brk{n_j(t,t-s)=k}$ for $k~=~1\ldots N$, where the dependence on $t$ is omitted.
    Then, from the birth-death Markov process properties, there holds
    \begin{align*}
        \frac{d}{ds}\vect P(s) = A\vect P(s),
    \end{align*}
    with $A$ the transition rate matrix of the process defined as the $N \times N$ tridiagonal matrix such that
    \begin{itemize}
        \item $\Melem{A}{k}{k} = -\prt{\frac{k(N-k)}{N-1}\lc +(k-1)\lr}$;
        \item $\Melem{A}{k}{k+1} = k\lr$;
        \item $\Melem{A}{k+1}{k} = \frac{k(N-k)}{N-1}\lc$.
    \end{itemize}
    
    The solution of this ODE system is given by
    \begin{align*}
        &\vect P(s) = e^{As}\vect P(0)&
        &\hbox{with}&
        &\vect P(0) = \vect e_1.
    \end{align*}
    
    Re-injecting the above expression into that of $\GossipFjit(s)$ yields
    \begin{align*}
        \GossipFjit(s) = w^Te^{As}\vect e_1,
    \end{align*}
    where $\Velem{w}{k} = \frac{k-1}{N-1}$.
    One has then
    \begin{align*}
        \Gossipfjit(s) = \frac{d}{ds}\GossipFjit(s) = w^TAe^{As}\vect e_1,
    \end{align*}
    which concludes the proof.
\end{proof}

It is interesting to notice that this is exactly the behavior of an SIS infection process (Susceptible-Infectious-Susceptible) \cite{MISC:Epidemiology}, where the disease is an information about a given agent.
New infections occur when an agent knowing that information interacts with one that does not, and healing happens through replacements of agents that know the information.

Since the pseudo-PDF $\Gossipfjit(s)$ obtained in Proposition~\ref{prop:SIS:fji} does not depend either on the agents nor the time $t$, the lower bound on $\Ep{C(t)}$ can be derived by applying Corollary~\ref{cor:General:symmetric_steadyState}.

\begin{thm}
\label{thm:SIS:bound}
For any algorithm that can be computed on $\GossipKset{i}(t)$, there holds
\begin{equation}
    \label{eq:thm:SIS:bound}
    \small
    \Ep{C(t)} 
    \geq \frac{N-1}{N^2} \prt{1-w^TA\prt{A-2\lr}^{-1}\brk{e^{\prt{A-2\lr}t}-I}\mathbf{e_1}} \sigma^2,
\end{equation}
where $w$ and $A$ are defined in Proposition~\ref{prop:SIS:fji}.
This holds even if agents know the system size $N$, the rates $\lr$ and $\lc$, and the distribution defining their values.
\end{thm}

\begin{proof}
We inject $\Gossipfjit(s)$ from Proposition~\ref{prop:SIS:fji} into equation (\ref{eq:cor:General:symmetric_steadyState:symmetry}) from Corollary~\ref{cor:General:symmetric_steadyState} in order to derive the bound. 
The integral then becomes
\begin{align*}
    \int_0^t \Gossipfjit(s)e^{-2\lr s}\mathrm{d}s
    &= \int_0^t w^TAe^{As}\mathbf e_1e^{-2\lr s}\mathrm{d}s.
\end{align*}

Using the commutativity between $A$ and $2\lr$, there holds
\begin{align*}
    \int_0^t w^TAe^{As}\mathbf e_1e^{-2\lr s}\mathrm{d}s
    &= w^TA\int_0^t e^{(A-2\lr)s} \mathrm ds \ \mathbf e_1.
\end{align*}

One can show that as long as $\lr\neq 0$, then $(A-2\lr)$ is invertible, and one has
\begin{align*}
    \int_0^t e^{(A-2\lr)s} \mathrm ds
    &= (A-2\lr)^{-1} \brk{e^{(A-2\lr)t}-I}.
\end{align*}

Injecting that result into (\ref{eq:cor:General:symmetric_steadyState:symmetry}) concludes the proof.
\end{proof}

Taking the limit as $t\rightarrow\infty$, or similarly by applying the second result of Corollary~\ref{cor:General:symmetric_steadyState}, we get the following corollary, that is proved in Appendix~\ref{sec:Annex:proof_Corollary_XYZ}.

\begin{cor}
    \label{cor:SIS:steadyState}
    For any algorithm that can be implemented on $\GossipKset{i}(t)$, there holds
    \begin{align}
        \label{eq:cor:SIS:symmetric_steadyState:symmetry_steadyState}
        \lim\inf_{t\rightarrow\infty} \Ep{C(t)} 
        &\geq \frac{N-1}{N^2} \prt{1-w^TA\prt{2\lr-A}^{-1}\mathbf{e_1}} \sigma^2\\
        \label{eq:cor:SIS:XYZ_bound}
        &\geq \frac{N-1}{N^2} \prt{\frac72 + \log\prt{\frac{N-2}{2}} + C_N} \rratio^{-1}\sigma^2\nonumber\\
        &\ \ \ \ \ \ + \mathcal O\prt{\frac{N}{\rratio^2}}\sigma^2,
    \end{align}
    where $C_N$ is some known polynomial term in $\mathcal O(N^{-1})$.
\end{cor}

Expression \eqref{eq:cor:SIS:XYZ_bound} is an analytical approximation of \eqref{eq:cor:SIS:symmetric_steadyState:symmetry_steadyState}, proved in Appendix~\ref{sec:Annex:proof_Corollary_XYZ}.
It gets more accurate as the rate ratio $\rratio = \lc/\lr$ increases, and as the size of the system $N$ decreases.
Ultimately, as the approximation gets more accurate, it becomes a lower bound on \eqref{eq:cor:SIS:symmetric_steadyState:symmetry_steadyState}, and thus a valid bound for the performance of algorithms implementable with the Gossip model.

\begin{figure}
    \centering
    \includegraphics[width=0.5\textwidth,clip = true, trim=0.5cm 10.25cm 0cm 10.5cm,keepaspectratio]{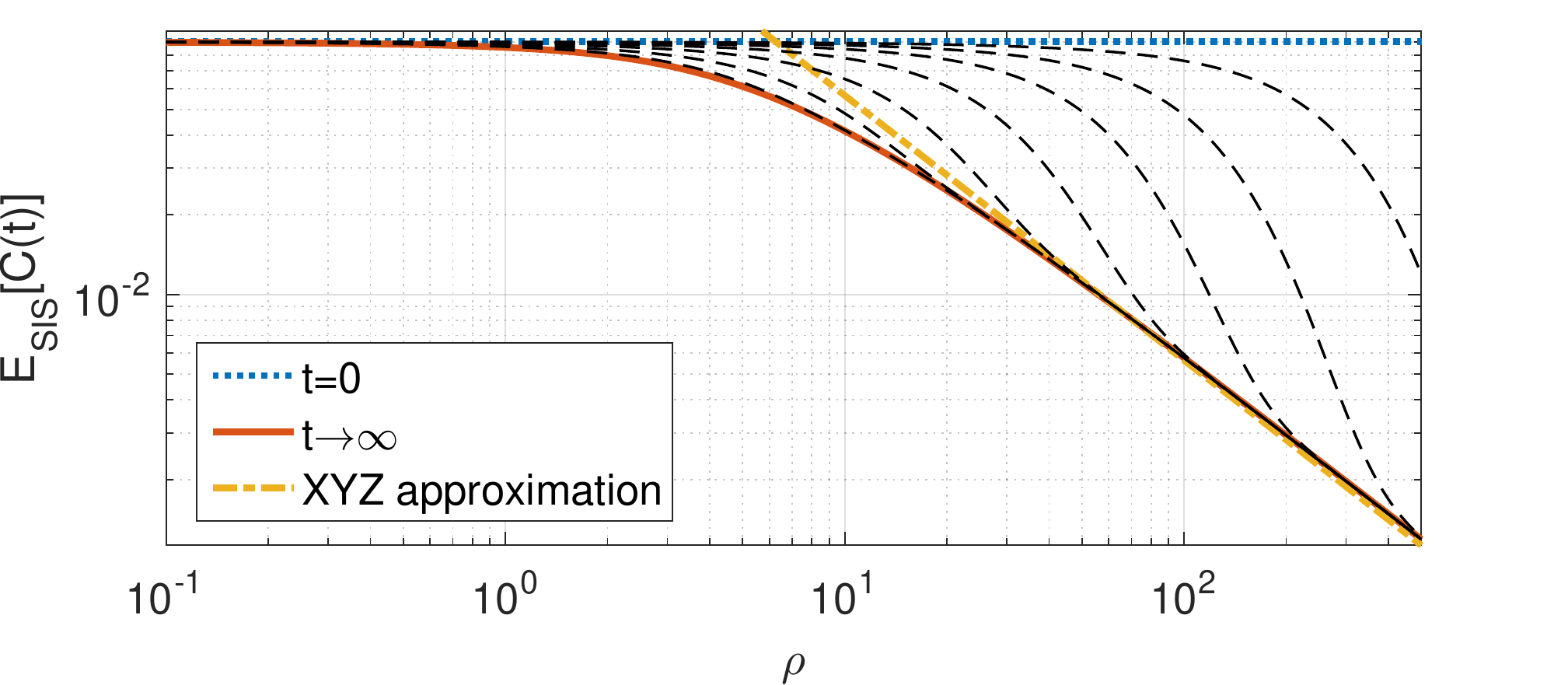}
    \caption{Time-dependent bound derived with the Gossip model (\ref{eq:thm:SIS:bound}) with 10 agents and with $\sigma^2=1$. The dotted blue line is the MSE at the initialization of the system, the plain red line is bound (\ref{eq:cor:SIS:symmetric_steadyState:symmetry_steadyState}) obtained when $t\rightarrow\infty$, the thick yellow dash-dotted line is the approximation (\ref{eq:cor:SIS:XYZ_bound}) of that bound, and the black dashed lines are the bounds at some times in between.} 
    \label{fig:SIS:TimeDepBound_XYZ_200}
\end{figure}

Fig.~\ref{fig:SIS:TimeDepBound_XYZ_200} shows the bound (\ref{eq:thm:SIS:bound}) in terms of the rate ratio $\rratio = \lc/\lr$ (\textit{i.e.} the expected number of interactions experienced by an agent before leaving the system) in the same settings as for the Ping model: 10 agents at different times from initialization until convergence to a time-independent bound as $t\rightarrow\infty$.
The same preliminary observations as for the Ping model hold.
The MSE is $\frac{N-1}{N^2}\sigma^2$ at the initialization, and remains around that value for small values of $\rho$. 
By contrast, the MSE decays to zero as $\rratio$ gets large and as the behavior of the system becomes that of a closed system.

However, in opposition with what was observed with the Ping model, the system size has here also an impact on the decay of the MSE.
This can be due to the propagation time of the information within the system, and at some extent to the memory losses at replacements, which are considered with the Gossip model and not with the Ping model.
Hence, the bound tends to be less conservative than with the Ping model.

We also show the approximation \eqref{eq:cor:SIS:XYZ_bound} of the bound from Corollary~\ref{cor:SIS:steadyState}, that is expected to get more accurate for large values of $\lc/\lr$ when $t\rightarrow\infty$.
With $10$ agents, it appears that it starts getting satisfyingly accurate around $\lc/\lr \approx 50$.

%%%%%%%%%%%%%%%%%%%%%%%%%%%%%%%%%%%%%%%%%%%%%%%%%%%%%%%%%%%%%%%%%%%%
\subsection{Performance analysis of the Gossip algorithm}
\label{Sec:App:Gossip}
We now compare the bounds we derived in the two previous sections with the performance of a specific algorithm, namely the Gossip algorithm \cite{Avg:Gossip}: an agent $i$ takes its own value $x_i$ as initial estimate $y_i(t)$ at its arrival in the system, and
each time two agents $i$ and $j$ interact, their estimates update as
\begin{equation}
    \label{eq:App:Gossip:Gossip_interaction}
    y_i(t^+) = y_j(t^+) = \frac{y_i(t^-)+y_j(t^-)}{2}.
\end{equation}

Since the Gossip algorithm rely on pairwise interactions, it can by definition be implemented on $\GossipKset{i}(t)$.
Similarly, the Gossip algorithm can be implemented on one instantiation of the Ping model.
Let us imagine that whenever two agents interact, they both perform a Ping update instead of performing a pairwise exchange of information; this would allow them to implement the Gossip algorithm.
Moreover, it corresponds to an instance of the Ping model with a specific dependency between the Poisson clocks driving the Ping updates, which is allowed by Definition~\ref{def:App:Ping:model}.
In particular, the Ping update rate would be $\lambda_p=\lambda_c$.
This can be shown using Lemma~\ref{lem:Stt:Kset:Kset_Inclusion}.
The bounds derived with those models are thus valid for the performance of the Gossip algorithm.
To avoid confusion, we refer to the second bound (\textit{i.e.}, pairwise interactions model) as the \quotes{\emph{SIS bound}} in the remainder of this section.

\begin{figure}
    \centering
    \includegraphics[width=0.5\textwidth,clip = true, trim=4.5cm 6.75cm 5cm 6.5cm,keepaspectratio]{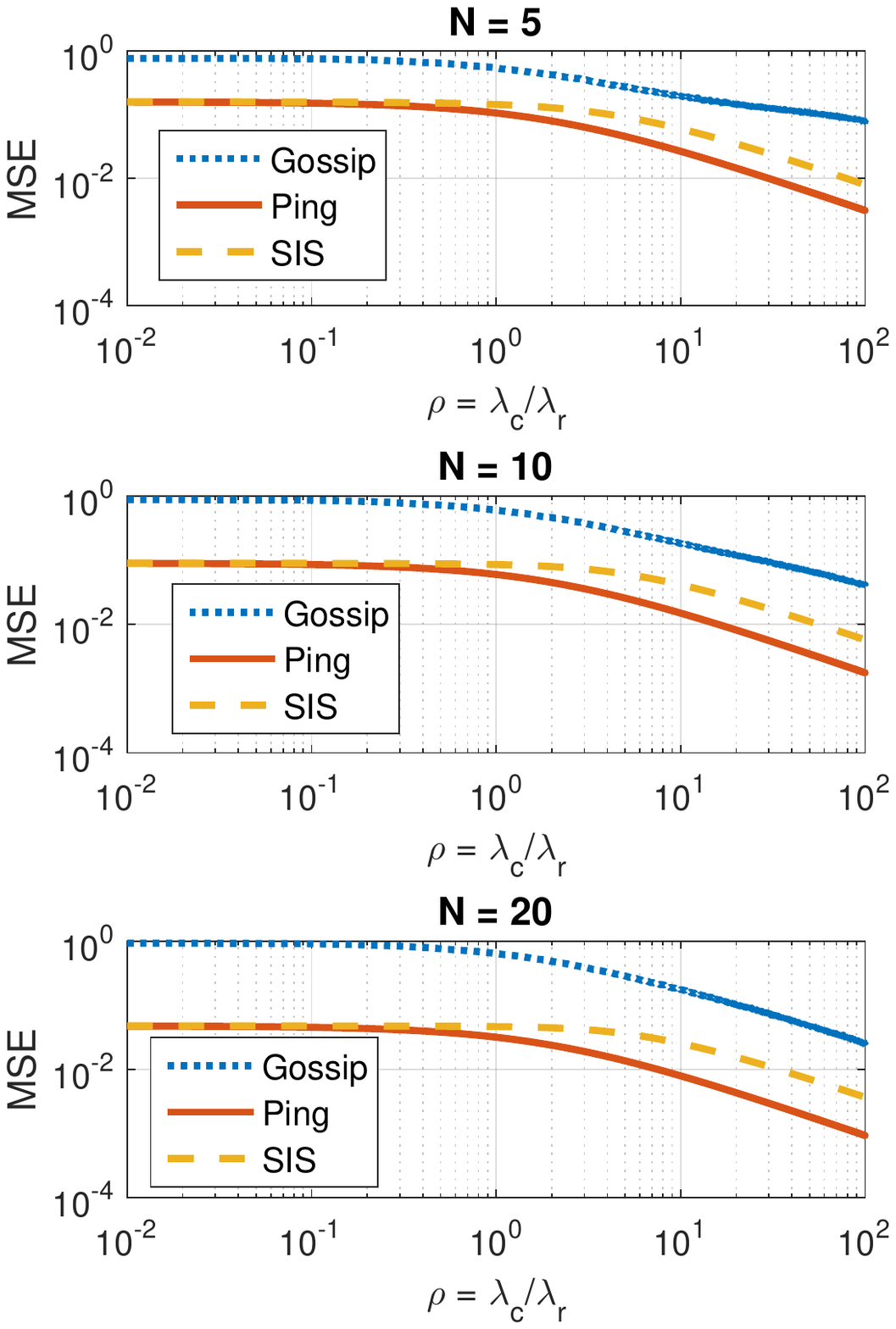}
    \caption{Performance of the Gossip algorithm (in dotted blue line) obtained through simulations compared with both the Ping bound (in plain red line) and the SIS bound (in yellow dashed line), for several values of the system size $N$. The simulation was performed over $200$ events (either replacements or communications) for $500$ realizations.}
    \label{fig:Gossip:N5N10N20}
\end{figure}

In Fig.~\ref{fig:Gossip:N5N10N20}, we compare both the asymptotic Ping and SIS bounds with the asymptotic performance of the simulated Gossip algorithm, respectively for $5$, $10$ and $20$ agents.
The bounds that are depicted are thus the bounds (\ref{eq:cor:Ping:symmetric_steadyState:symmetry_steadyState}) and \eqref{eq:cor:SIS:symmetric_steadyState:symmetry_steadyState}, respectively from Corollaries~\ref{cor:Ping:symmetric_steadyState} and \ref{cor:SIS:steadyState}.

As expected from the derivation of the bounds, it appears that the SIS bound \eqref{eq:thm:SIS:bound} is tighter than the Ping bound \eqref{eq:thm:Ping:bound}, since it is obtained with a more constrained interaction model.
More precisely, the SIS bound stands for the exact interaction model on which the Gossip algorithm is defined.
It is thus closer to the actual performance of the Gossip algorithm than the Ping bound.

Nevertheless, there is still an important gap between the performance depicted by both bounds and that observed for the simulated Gossip algorithm.
Interestingly, the size of the system has barely no impact on the performance of the Gossip algorithm, especially small values of $\rratio=\lc/\lr$.
In opposition, the bounds depict a smaller MSE as the system size increases.
This is highlighted by Fig.~\ref{fig:SISnGossip_wrtN}, which shows the SIS bound and the performance of the Gossip algorithm with respect to the size of the system $N$ for several values of $\rratio$.
This gives insight on how knowing the size of the system (which is assumed to be the case for the bounds, but is not for the Gossip algorithm) can impact the performance of algorithms.
Indeed, the MSE depicted by the bounds is scaled using the size of the system, whereas it has no influence on the MSE of the Gossip algorithm.

\begin{figure}
    \centering
    \includegraphics[width=0.5\textwidth]{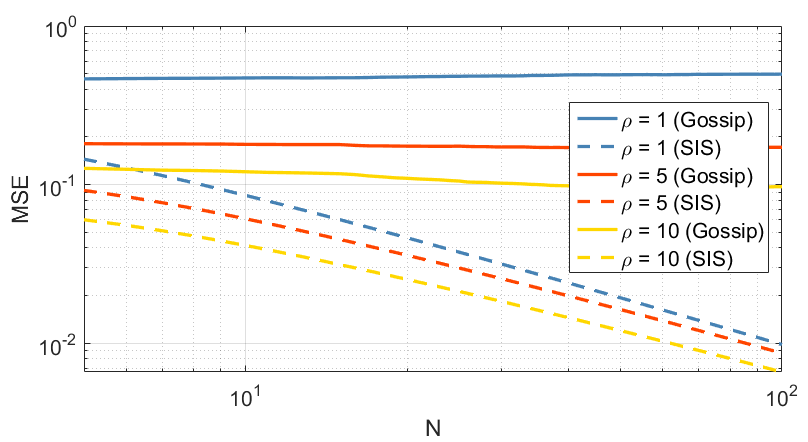}
    \caption{Performance of the Gossip algorithm (in plain line) and of the SIS bound \eqref{eq:thm:SIS:bound} (in dashed line) with respect to the size of the system $N$ for several values of the rate ratio $\rho$.}
    \label{fig:SISnGossip_wrtN}
\end{figure}

Nonetheless, the behavior of the performance of the Gossip algorithm is qualitatively well captured by both bounds.
This is surprising, especially because the Gossip algorithm is a particularly naive algorithm: it only relies on one variable, and does not make use of any identifier nor any of the information related to the openness of the system or to the distribution of the intrinsic values of the agents, by contrast with our bounds.
This questions the interest of such parameters in the design of efficient algorithms, and the exact impact they have on their performance.

%%%%%%%%%%%%%%%%%%%%%%%%%%%%%%%%%%%%%%%%%%%%%%%%%%%%%%%%%%%%%%%%%%%%
%%%%%%%%%%%%%%%%%%%%%%%%%%%%%%%%%%%%%%%%%%%%%%%%%%%%%%%%%%%%%%%%%%%%
\section{Conclusion}
\label{Sec:Ccl}
In this work, we analyzed open multi-agent systems, where agents can join and leave the system during the process.
We highlighted several challenges arising from that property, in particular the impossibility for algorithms to converge due to the variations of size, state and objective experienced by the system, and to the time it takes for information to propagate.

We focused on the derivation of lower bounds on the performance of algorithms trying to achieve average consensus in open multi-agent systems of constant size.
Those were obtained by analyzing the Mean Squared Error of an algorithm that achieves provably optimal performance in our setting.
We defined \textit{knowledge sets} to manage the information made potentially available to an agent given a rule defining the interactions happening in the system.
It allowed deriving a general lower bound that can be instantiated by defining that rule.
We finally applied our result to derive two lower bounds that are valid for the Gossip algorithm, to analyze the gap between our bounds and the actual performance of an algorithm that does not make use of most of the knowledge on which our bounds rely nor on identifiers.

In this work, we focused on properly defining and handling the information in the system, which we consider as one of the big challenges with open systems.
Another important challenge consists in considering systems of variable size; this actually relies on additional tools and is left for future research.

This study aimed at setting up tools and a methodology for studying open systems in general, and those can thus be applied to more complex objectives, such as decentralized optimization \cite{DO:dualAvg} or formation control \cite{ApMas:Consensus_based_formation_control:CDC2019}, or to more structured or constrained interaction rules.
Investigating the impact of some parameters assumed to be known in our model on the actual performance of algorithms is also an interesting follow-up; in particular the impact of identifiers is questionable considering the performance of our bounds as compared to that of anonymous Gossip interactions.
Finally, extending our analysis to continuous-time algorithms reveals to be challenging, as our derivation relies on essentially discrete-time ideas.

%%%%%%%%%%%%%%%%%%%%%%%%%%%%%%%%%%%%%%%%%%%%%%%%%%%%%%%%%%%%%%%%%%%%
%%%%%%%%%%%%%%%%%%%%%%%%%%%%%%%%%%%%%%%%%%%%%%%%%%%%%%%%%%%%%%%%%%%%
\bibliographystyle{IEEEtran}
\bibliography{ITAC_FPL.bib}

%%%%%%%%%%%%%%%%%%%%%%%%%%%%%%%%%%%%%%%%%%%%%%%%%%%%%%%%%%%%%%%%%%%%
%%%%%%%%%%%%%%%%%%%%%%%%%%%%%%%%%%%%%%%%%%%%%%%%%%%%%%%%%%%%%%%%%%%%
\appendix
\subsection{Proof of Lemma~\ref{lem:General:single_error:MSE_RV}}
\label{sec:Annex:proof_Lemma_MSE_RV}

\begin{proof}
Let $Y$ and $Z$ be two i.i.d. zero mean random variables of variance $\sigma^2$, and define another random variable $X$ as follows
\begin{equation}
    \label{Annex:eq:lem:General:single_error:MSE_RV:rv_def}
    X := 
    \begin{cases}
        Z &\hbox{with probability } p\\
        Y &\hbox{with probability } 1-p
    \end{cases},
\end{equation}
with the event related to the probability $p$ independent of both $Z$ and $Y$.

By definition of $X$, the Mean Squared Error of an estimator $\hat X$ of $X$ provided $Z$ is given by the following
\begin{align*}
    &\condE{\prt{X-\hat X}^2}{Z} \\
    &\ \ \ = \condE{\prt{Z-\hat X}^2}{Z}p+\condE{\prt{Y-\hat X}^2}{Z}(1-p)\\
    &\ \ \ = p\prt{Z-\hat X}^2 + (1-p)  \hat X^2 + (1-p) \sigma^2\\
    &\ \ \ = pZ^2 - 2pZ\hat X +\hat X^2 +(1-p)\sigma^2.
\end{align*}
The random variable $Y$ is eliminated in the development above by using the fact that it is independent of $Z$, and that $\Ep{Y} = 0$ and $\Ep{Y^2} = \sigma^2$ by definition.

That estimator is thus optimal if
\begin{align*}
    \frac{d}{d\hat X}\condE{\prt{X-\hat X}^2}{Z}
    &= \frac{d}{d\hat X} \prt{pZ^2 - 2pZ\hat X +\hat X^2 +(1-p)\sigma^2}\\
    &= -2pZ+2\hat X = 0\\
    \Leftrightarrow \hat X &= pZ,
\end{align*}
which leads to the first result stated in (\ref{eq:lem:General:single_error:MSE_RV:estimate}).
    
The MSE for that estimator $\hat X$ provided $Z$ is then given by
\begin{align*}
    \condE{\prt{X-\hat X}^2}{Z}
    &= pZ^2 -2p^2Z^2 + p^2Z^2 + (1-p)\sigma^2\\
    &= p(1-p)Z^2 + (1-p)\sigma^2.
\end{align*}
 
This finally leads to 
\begin{align*}
    \Ep{\prt{X-\hat X}^2}
    &= p(1-p)\Ep{Z^2} + (1-p)\sigma^2\\
    &= (1-p)(1+p)\sigma^2 \\
    &= (1-p^2)\sigma^2,
\end{align*}
which concludes the proof.
\end{proof}

%%%%%%%%%%%%%%%%%%%%%%%%%%%%%%%%%%%%%%%%%%%%%%%%%%%%%%%%%%%%%%%%%
\subsection{Proof of Corollary~\ref{cor:SIS:steadyState}}
\label{sec:Annex:proof_Corollary_XYZ}
We first prove a few preliminary results that will allow us to build the proof.
Note that for this proof, we use the notation $*$ to refer to a whole row or column of a matrix (\textit{e.g.} $\Melem A i*$ stands for the whole $i$-th row of $A$).

We start from expression \eqref{eq:cor:SIS:symmetric_steadyState:symmetry_steadyState} that is directly obtained from taking the limit as $t\to\infty$ of \eqref{eq:thm:SIS:bound} from Theorem~\ref{thm:SIS:bound}.

\begin{equation}
    \label{eq:Annex:SIS:E}
    E := 1-w^TA\prt{2\lr-A}^{-1}\mathbf e_1,
\end{equation}
such that expression (\ref{eq:cor:SIS:symmetric_steadyState:symmetry_steadyState}) is equivalent to $\tfrac{N-1}{N}\sigma^2E$.

\begin{lem}
    \label{lem:Annex:SIS_rho_expression}
    The expression (\ref{eq:Annex:SIS:E}) is equivalent to the following:
    \begin{equation}
        \label{eq:lem:Annex:SIS_rho_expression}
        E = 1-w^T\prt{v-w\lclr^{-1}}^T\prt{\lclr^{-1}(2I-R)-C}^{-1}\mathbf e_1,
    \end{equation}
    where $\Velem{v}{k} = \tfrac{k(N-k)}{N-1}$, $\lclr := \frac{1}{N-1}\frac\lc\lr$, and such that $C\lclr+R = \tfrac{1}{\lr}A$.
\end{lem}

\begin{proof}
    Denote $\tilde A := A/\lr$, then one has
    \begin{equation*}
        E = 1 - w^T \tilde A \prt{2I-\tilde A}^{-1}\mathbf e_1.
    \end{equation*}
    Moreover, by definition of the matrix $A$, there holds that $$\tilde A = C\lclr+R,$$ where the matrices $C$ and $R$ are bidiagonal matrices defined as follows: 
    \begin{align*}
        &\Melem{C}{k}{k} = -\Melem{C}{k+1}{k} = -(N-k)& &(1\leq k \leq N-1);\\
        &\Melem{R}{k}{k} = -\Melem{R}{k-1}{k} = -(k-1)& &(2\leq k \leq N).
    \end{align*}

    Hence, one has 
    \begin{align*}
        E 
        &= 1-w^T\prt{C\lclr+R}\prt{2I-C\lclr-R}^{-1}\mathbf e_1\\
        &= 1-w^T\prt{C+R\lclr^{-1}}\prt{\lclr^{-1}(2I-R)-C}^{-1}\mathbf e_1.
    \end{align*}

    Moreover, we define the vector $v$ as follows
    \begin{align*}
        \Velem{v}{k} 
        &= \Velem{w^TC}{k} 
        = w^T\Melem{C}{*}{k} 
        = \Velem{v}{k} \Melem{C}{k}{k} + \Velem{v}{k+1} \Melem{C}{k+1}{k}\\
        &= \tfrac{k-1}{N-1}\bigprt{-k(N-k)} + \tfrac{k}{N-1}\bigprt{k(N-k)}
        = \tfrac{k(N-k)}{N-1},
    \end{align*}

    and we observe that $w^TR = -w^T$:
    \begin{align*}
        \Velem{w^TR}{k}
        &= w^T\Melem{R}{*}{k}
        = \Velem{v}{k} \Melem{R}{k}{k} + \Velem{v}{k-1} \Melem{R}{k-1}{k}\\
        &= -\tfrac{k-1}{N-1}(k-1) + \tfrac{k-2}{N-1}(k-1) = -\tfrac{k-1}{N-1} = \Velem{w}{k}.
    \end{align*}

    Finally the conclusion follows
    \begin{align*}
        E 
        &= 1-\prt{w^TC+w^TR\lclr^{-1}}\prt{\lclr^{-1}(2I-R)-C}^{-1}\mathbf e_1\\
        &= 1-\prt{v^T-w^T\lclr^{-1}}\prt{\lclr^{-1}(2I-R)-C}^{-1}\mathbf e_1.
    \end{align*}
\end{proof}

\begin{prop}
    \label{prop:Annex:SIS:XYZ_assumption}
    Considering expression (\ref{eq:lem:Annex:SIS_rho_expression}) from Lemma~\ref{lem:Annex:SIS_rho_expression}, there holds
    \begin{equation}
        \label{eq:prop:Annex:SIS:XYZ_assumption}
        \prt{\lclr^{-1}(2I-R)-C}^{-1} = X\lclr + Y + Z\lclr^{-1} +\mathcal O\prt{\lclr^{-2}},
    \end{equation}
    for some matrices $X, Y, Z \in \R^{N\times N}$.
\end{prop}

\begin{proof}
    Denote $\lrlc := \lclr^{-1}$, and $M := (2I-R)\lrlc-C$ and define 
    \begin{equation*}
        S_{ij} := 
        \begin{bmatrix}
        M               &\vdots &\mathbf e_i\\
        \ldots          &       &\ldots\\ 
        \mathbf e_j^T   &\vdots &0
        \end{bmatrix}.
    \end{equation*}

    Then $-\Melem{M^{-1}}{j}{i} = -\mathbf e_j^T M \mathbf e_i$ is the Schur complement of $M$ in $S_{ij}$. 
    Hence, from the properties of the Schur complement, there holds
    \begin{equation*}
        \Melem{M^{-1}}{j}{i} = \det\prt{\mathbf e_j^T M \mathbf e_i} = -\frac{\det\prt{S_{ij}}}{\det(M)}.
    \end{equation*}

    Let $\det(M)$ and $\det(S_{ij})$ be polynomials such that
    \begin{align*}
    \label{eq:proof:prop:Annex:SIS:XYZ_assumption:q}
        \det(M) 
        &= \lrlc^q p(\lrlc)&
        &\hbox{with } p(0)\neq 0;\\
        \det(S_{ij})
        &= \lrlc^{q_{ij}} p_{ij}(\lrlc)&
        &\hbox{with } p_{ij}(0)\neq 0 \ \ \ \forall i,j.
    \end{align*}

    Hence,
    \begin{equation*}
        \Melem{M^{-1}}{j}{i} 
        = c_0 \lrlc^{q-q_{ij}} + c_1 \lrlc^{q-q_{ij}+1} + c_2 \lrlc^{q-q_{ij}+2} + \ldots;
    \end{equation*}

    and thus
    \begin{equation*}
        M^{-1}
        = X \lrlc^{q-\tilde q_{ij}} + Y \lrlc^{q-\tilde q_{ij}+1} + Z \lrlc^{q-\tilde q_{ij}+2} + \ldots,
    \end{equation*}
    with $\tilde q_{ij} := \min_{i,j} \brc{q_{ij}}$.

    By definition of $M$, and from the properties of tridiagonal matrices, there holds
    \begin{align*}
        \det(M_k) 
        = &\brk{k(N-k)+(k+1)\lrlc} \det(M_{k-1})\\
          &-(k-1)^2(N-k+1)\lrlc \det (M_{k-2}),
    \end{align*}
    with $\det(M_0)=1$ and $\det(M_{-1})=0$, and where $M_k$ denotes the $k\times k$ upper-left sub-matrix of $M$.
    We show by induction that $\det(M_k) = p_k(\lrlc)$ with $p_k(\lrlc)$ a polynomial such that $p_k(0)\neq0$ for $k<N$.
    \begin{itemize}
        \item \underline{Initial case:}\\
        $k=0 \Rightarrow \det(M_0) = 1$\\
        $k=1 \Rightarrow \det(M_1) = (N-1) + 2\lrlc$
    
        \item \underline{Assume it is true for $k-1$ and $k$:}
        \begin{align*}
            \det(M_{k+1}) 
            &= (k+1)(N-k+1)p_{k}(\lrlc)\\
            &\ \ \ + (k+2)\lrlc p_k(\lrlc)- k^2(N-k)\lrlc p_{k-1}(\lrlc)\\
            &= p_{k+1}(\lrlc) \ \ \ \ \hbox{such that } p_{k+1}(0)\neq0. 
        \end{align*}
    \end{itemize}

    Hence, one obtains the following for $k=N$:
    \begin{align*}
        \det(M_N) 
        &= (N+1)\lrlc p_{N-1}(\lrlc) - (N-1)^2\lrlc p_{N-2}(\lrlc)\\
        &= \lrlc p(\lrlc) \ \ \ \hbox{with } p(0)\neq0,
    \end{align*}
    and it follows that $q=1$ (with $q$ defined previously).

    Moreover, from the structure of $M$, the determinant of any cofactor of $M$ is a polynomial of $\lrlc$, and for all $i,j$, one has $q_{ij}\geq0$.

    Hence,
    \begin{align*}
        M^{-1}
        &= X\lrlc^{-1} + Y + Z\lrlc + \ldots\\
        &= X\lclr + Y + Z\lclr^{-1} + \mathcal O(\lclr^{-2}),
    \end{align*}
    which concludes the proof.
\end{proof}

\begin{prop}
    \label{prop:Annex:SIS:ALG_XYZbound}
    Expression (\ref{eq:Annex:SIS:E}) reduces to the following:
    \begin{equation}
        \label{eq:prop:Annex:SIS:ALG_XYZbound}
        E = \frac{\lclr^{-1}}{N-1} \brk{\sum_{k=2}^{N-2} \prt{\frac{N}{k(N-k)}} + \frac{3N-1}{N-1}} + \mathcal O \prt{\lclr^{-2}}. 
    \end{equation}
\end{prop}

\begin{proof}
    We identify the matrices $X$, $Y$ and $Z$ from Proposition~\ref{prop:Annex:SIS:XYZ_assumption} by imposing $\lclr^k\approx0$ for all $k\leq-2$, and
    \begin{align*}
        I 
        &= \prt{X\lclr+Y+Z\lclr^{-1}} \prt{\lclr^{-1}(2I-R)-C}\\
        &= \prt{\lclr^{-1}(2I-R)-C} \prt{X\lclr+Y+Z\lclr^{-1}}.
    \end{align*}

    This leads to six equations allowing identifying $X$, $Y$ and $Z$.
    For concision matters, we do not detail the whole algebra behind the identification.

    \begin{enumerate}[(i)]
        \item \underline{$CX = 0$ and $XC = 0$:}\\
        The identification gives $X = \alpha\mathbf e_N\ONE^T$ for some $\alpha\in\R$.
    
        \item \underline{$(2I-R)X-I = CY$:}\\
        This equality provides $\alpha = \tfrac12 \Rightarrow X = \tfrac12\mathbf e_N\ONE^T$.
        Moreover, it gives 
        \begin{equation*}
            \Melem{Y}{k}{*} = 
            \begin{cases}
                \frac{\sum\nolimits_{j=1}^k \mathbf e_j^T}{k(N-k)}  &k<N-1\\
                \frac{N+1}{2(N-1)}\ONE^T -\frac{\mathbf e_N^T}{N-1} &k=N-1.
            \end{cases}
        \end{equation*}
    
        \item \underline{$(2I-R)X-I = YC$:}\\
        The identification shows that $\Melem{Y}{N}{*} = \beta + \sum_{j=1}^{k-1} \frac{1}{j(N-j)}$ for some $\beta \in \R$.
    
        \item \underline{$(2I-R)Y-CZ = 0$:}\\
        This equality provides $\beta = -\tfrac12 - \sum\nolimits_{j=1}^{N-1} \tfrac{1}{j(N-j)}$.\\
        Also, the identification gives, for $1\leq k\leq N-1$:
        \begin{equation*}
            \Melem{Z}{k}{*} = \tfrac{1}{k(N-k)}\prt{k\Melem{Y}{k+1}{*} - \sum\nolimits_{j=1}^k \Melem{Y}{j}{*}}.
        \end{equation*}
    
        \item \underline{$Y(2I-R)-ZC = 0$:}\\
        Finally, the last equality allows identifying the first column of the matrix $Z$, which we will see later is sufficient to entirely characterize the wanted expression: 
        \begin{equation*}
            \small
            \Melem{Z}{k}{1} = 
            \begin{cases}
                \tfrac{1}{k(N-k)} \prt{\tfrac{k}{(k+1)(N-k-1)} - 2\sum_{j=1}^k \tfrac{1}{j(N-j)}} &k<N-2\\
                \tfrac{1}{2(N-2)} \prt{\tfrac{(N-2)(N+1)}{2(N-1)} - 2 \sum_{j=1}^{N-2} \tfrac{1}{j(N-j)}} &k=N-2\\
                \tfrac{-1}{N-1} \prt{\tfrac{N+1}{2} + (N+1)\sum_{j=1}^{N-1}\tfrac{1}{j(N-j)}}   &k=N-1.
            \end{cases}
        \end{equation*}
    \end{enumerate}

    We can then compute $E$ from expression (\ref{eq:lem:Annex:SIS_rho_expression}) using the approximation from Proposition~\ref{prop:Annex:SIS:XYZ_assumption} with the identified matrices $X$, $Y$ and $Z$.

    Expression (\ref{eq:lem:Annex:SIS_rho_expression}) using Proposition~\ref{prop:Annex:SIS:XYZ_assumption} reduces to
    \begin{align*}
        E
        &= 1 - v^TX\mathbf e_1\lclr - v^TY\mathbf e_1 - v^TZ\mathbf e_1 \lclr^{-1}\\
        &\ \ \ + w^TX\mathbf e_1 + w^T Y \mathbf e_1 \lclr^{-1} + \mathcal O \prt{\lclr^{-2}}.
    \end{align*}
    We identify each term using the identification of $X$, $Y$ and $Z$:
    \begin{itemize}
        \item $v^TX\mathbf e_1 \lclr = v^T\Melem{X}{*}{1}\lclr = \tfrac12\Velem{v}{N}\lclr = 0$;\\
    
        \item $v^TY\mathbf e_1 = v^T\Melem{Y}{*}{1} = 3/2$;\\ 
    
        \item $w^TX\mathbf e_1 = w^T\Melem{X}{*}{1} = 1/2$;\\
    
        \item $w^TY\mathbf e_1\lclr^{-1} = \frac{-1}{N-1}\prt{1/2 + \sum\nolimits_{k=1}^{N-1}(1/k)}\lclr^{-1}:$%\\   
        \begin{align*}
            w^TY\mathbf e_1 
            &= w^T\Melem{Y}{*}{1}
            = \sum\nolimits_{k=1}^N = \tfrac{k-1}{N-1}\Melem{Y}{k}{1}\\
            &= \sum\nolimits_{k=1}^{N-1}\tfrac{k-1}{N-1} \tfrac{1}{k(N-k)} + \tfrac{N-2}{2(N-1)} - \tfrac12 - \sum\nolimits_{k=1}^{N-1}\tfrac{1}{k(N-k)}\\ 
            &= \frac{-1}{N-1}\prt{\frac12 + \sum_{k=1}^{N-1}(1/k)}.
        \end{align*}\\

        \item $v^TZ\mathbf e_1 \lclr^{-1}= v^T\Melem{Z}{*}{1} \lclr^{-1} = \sum_{k=1}^{N-1} \tfrac{k(N-k)}{N-1}\Melem{Z}{k}{1}\lclr^{-1}.$
    
        For the sake of concision, we do not detail the algebraic steps. Using the identified first column of the matrix $Z$, and reducing the expression leads to
        \begin{equation*}
            v^TZ\mathbf e_1\lclr^{-1} = \tfrac{1}{N-1} \prt{-\sum_{k=2}^{N-2} \prt{\tfrac{N+k}{k(N-k)}} - \tfrac32\tfrac{3N-1}{N-1} } \lclr^{-1}.
        \end{equation*}
    \end{itemize}

    Combining every term together to retrieve the expression of $E$ gives then
    \begin{align*}
        E
        &= 1 -\tfrac32 + \tfrac12 - 0\lclr + \prt{w^TY\mathbf e_1-v^TZ\mathbf e_1}\lclr^{-1} + \mathcal O\prt{\lclr^{-2}}\\
        &= \frac{\lclr^{-1}}{N-1} \prt{\sum_{k=2}^{N-2}\prt{\tfrac{N+k}{k(N-k)}} + \tfrac32\tfrac{3N-1}{N-1} -\tfrac12 - \sum_{k=1}^{N-1}\tfrac1k } + \mathcal O\prt{\lclr^{-2}}\\
        &= \frac{\lclr^{-1}}{N-1} \prt{\sum_{k=2}^{N-2}\prt{\tfrac{N}{k(N-k)}} + \tfrac32\tfrac{3N-1}{N-1} -\tfrac32 - \tfrac{1}{N-1} } + \mathcal O\prt{\lclr^{-2}}\\
        &= \frac{\lclr^{-1}}{N-1} \prt{\sum_{k=2}^{N-2} \prt{\tfrac{N}{k(N-k)}} + \tfrac{3N-1}{N-1}} + \mathcal O \prt{\lclr^{-2}},
    \end{align*}
    which concludes the proof.
\end{proof}

\begin{proof}[Proof of Corollary~\ref{cor:SIS:steadyState}]
We can finally provide the proof for expression \eqref{eq:cor:SIS:XYZ_bound} from Corollary~\ref{cor:SIS:steadyState}. 
We start from the result of Proposition~\ref{prop:Annex:SIS:ALG_XYZbound}, and bound it from below.
The composite trapezoid rule for integrating convex functions allows writing
\begin{equation*}
    \int_2^{N-2}\frac{1}{x(N-x)}\mathrm dx \leq \sum_{k=2}^{N-3}\prt{\frac{1}{k(N-k)}+\frac{1}{(k+1)(N-k-1)}},
\end{equation*}
where the inequality is guaranteed by the convexity of the function $\frac{1}{x(N-x)}$.
Hence, it follows that 
\begin{equation*}
    \sum_{k=2}^{N-2}\frac{1}{k(N-k)} \geq \int_2^{N-2}\frac{1}{x(N-x)}\mathrm dx + \frac{1}{2(N-2)},
\end{equation*}
where the last term compensates those lacking from the trapezoid rule.

Moreover, one can compute the integral as follows, using partial fractions decomposition
\begin{align*}
    \int_2^{N-2}\frac{1}{x(N-x)}\mathrm dx
    &= \int_2^{N-2}\frac{1}{Nx}\mathrm dx + \int_2^{N-2}\frac{1}{N(N-x)}\mathrm dx\\
    &= \frac1N\bigbrk{\log(x) - \log(N-x)}_2^{N-2}\\
    &= \frac2N\log\prt{\frac{N-2}{2}}.
\end{align*}

Injecting the above inequality into expression (\ref{eq:prop:Annex:SIS:ALG_XYZbound}) from Proposition~\ref{prop:Annex:SIS:ALG_XYZbound} gives
\begin{align*}
    E 
    &\geq \frac{\lclr}{N-1} \brk{\frac{3N-1}{N-1} + \frac{N}{2(N-2)} + 2\log\prt{\frac{N-2}{2}}} + \mathcal O\prt{\lclr^{-2}}\\
    &\geq \frac{\lclr}{N-1} \brk{\frac72 + 2\log\prt{\frac{N-2}{2}} + \frac{N-5}{(N-2)(N-1)}} + \mathcal O\prt{\lclr^{-2}}.
\end{align*}
Injecting this expression into (\ref{eq:cor:SIS:symmetric_steadyState:symmetry_steadyState}) concludes the proof.
\end{proof}

%%%%%%%%%%%%%%%%%%%%%%%%%%%%%%%%%%%%%%%%%%%%%%%%%%%%%%%%%%%%%%%%%%%%
%%%%%%%%%%%%%%%%%%%%%%%%%%%%%%%%%%%%%%%%%%%%%%%%%%%%%%%%%%%%%%%%%%%%
\begin{IEEEbiography}[{\includegraphics[width=1in,height=1.25in,clip,keepaspectratio]{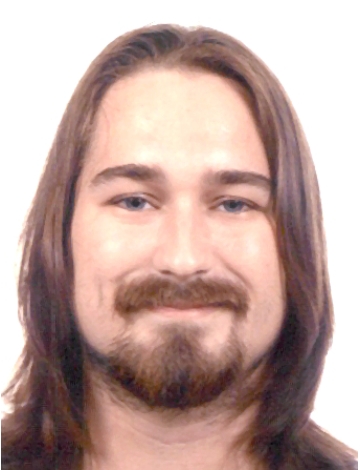}}]{Charles Monnoyer de Galland}
is a PhD student at UCLouvain, in the ICTEAM Institute, under the supervision of Professor Julien M. Hendrickx since 2018.
He is a FRIA fellow (F.R.S.-FNRS).

He obtained an engineering degree in applied mathematics (2018), and started a PhD in mathematical engineering the same year at the same university.
His research interests are centered around the analysis of open multi-agent systems.
\end{IEEEbiography}

\begin{IEEEbiography}[{\includegraphics[width=1in,height=1.25in,clip,keepaspectratio]{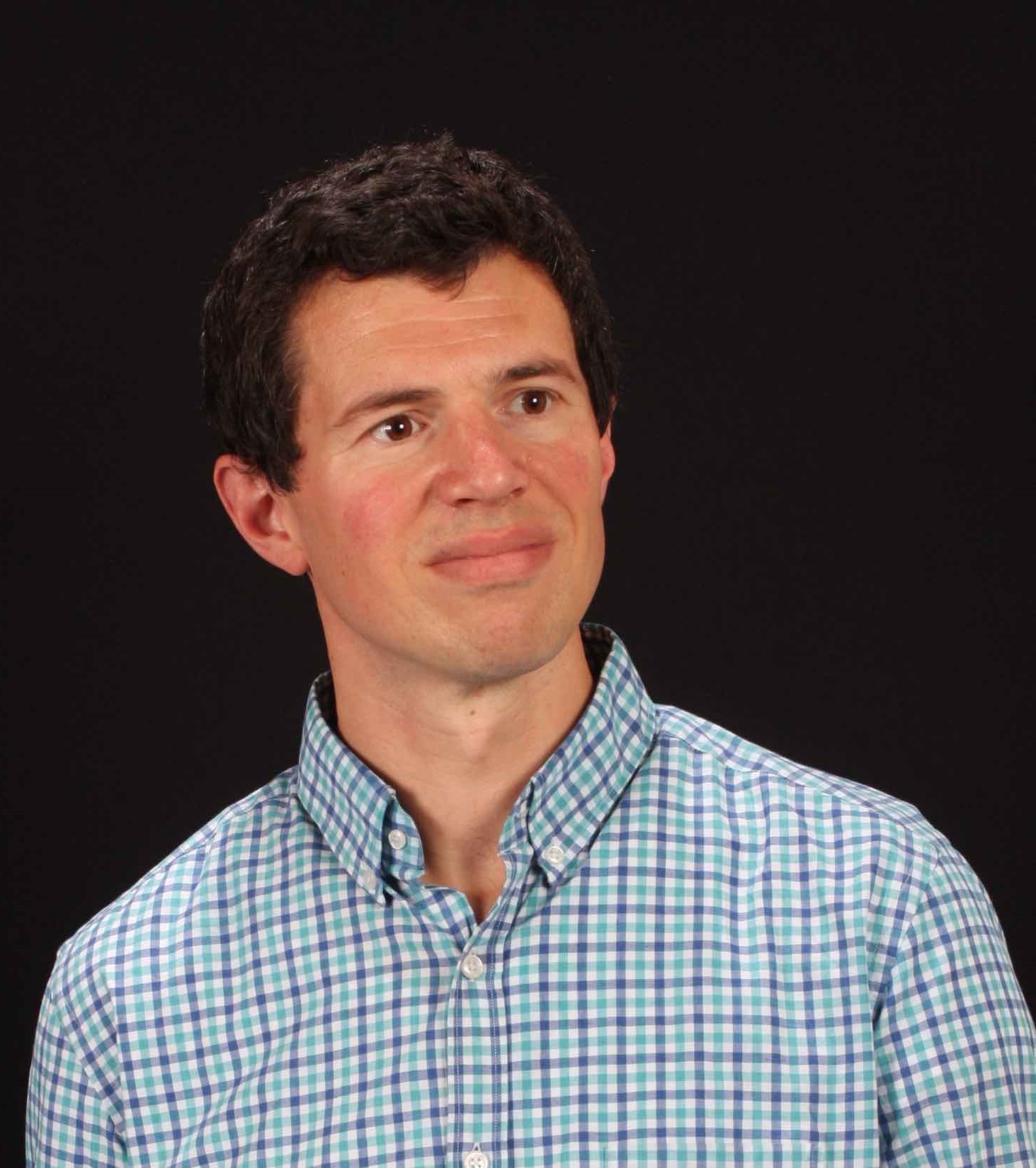}}]{Julien M. Hendrickx}
is professor of mathematical engineering at UCLouvain, in the Ecole Polytechnique de Louvain since 2010.

He obtained an engineering degree in applied mathematics (2004) and a PhD in mathematical engineering (2008) from the same university. He has been a visiting researcher at the University of Illinois at Urbana Champaign in 2003-2004, at the National ICT Australia in 2005 and 2006, and at the Massachusetts Institute of Technology in 2006 and 2008. He was a postdoctoral fellow at the Laboratory for Information and Decision Systems of the Massachusetts Institute of Technology 2009 and 2010, holding postdoctoral fellowships of the F.R.S.-FNRS (Fund for Scientific Research) and of Belgian American Education Foundation. He was also resident scholar at the Center for Information and Systems Engineering (Boston University) in 2018-2019, holding a WBI.World excellence fellowship.

Doctor Hendrickx is the recipient of the 2008 EECI award for the best PhD thesis in Europe in the field of Embedded and Networked Control, and of the Alcatel-Lucent-Bell 2009 award for a PhD thesis on original new concepts or application in the domain of information or communication technologies. 
\end{IEEEbiography}

\end{document}